\documentclass[runningheads,a4paper]{llncs}

\pdfoutput=1

\usepackage[data]{logic9-thm}
\usepackage{tikz}
\usetikzlibrary{automata,fit,positioning,arrows}
\usetikzlibrary{patterns}
\usetikzlibrary{decorations,decorations.pathmorphing,decorations.markings}

\usepackage{listings, multicol,mathtools}

\lstdefinelanguage{algo}{%
  morekeywords={function,procedure,algorithm,push,pop,top,for,all,and,or,if,then,else,repeat,until,while,do,report,return,such,that,each,add,call,exit,let}
}

\usepackage{todonotes}

\newcommand{\Lnew}{L_{new}}
\newcommand{\Unew}{U_{new}}
\newcommand{\alu}{\mathfrak{a}_{\fleq \scriptscriptstyle{LU} }}
\newcommand{\aluP}{\mathfrak{a}_{\fleq \scriptscriptstyle{L'U'} }}
\newcommand{\aluOne}{\mathfrak{a}_{\fleq \scriptscriptstyle{L_1U_1} }}
\newcommand{\aluOneP}{\mathfrak{a}_{\fleq \scriptscriptstyle{L_1'U_1'} }}

\newcommand{\aluTwoP}{\mathfrak{a}_{\fleq \scriptscriptstyle{L'_2U'_2} }}
\newcommand{\aluNew}{\mathfrak{a}_{\fleq \scriptscriptstyle{L_{new}U_{new}} }}
\newcommand{\ALU}[1]{\mathfrak{a}_{\fleq \scriptscriptstyle{#1} }}

\newcommand{\post}{\operatorname{Post}}
\newcommand{\pre}{\operatorname{Pre}}

\newcommand{\luNew}{\fleq_{\scriptscriptstyle \Lnew\Unew}}
\newcommand{\luP}{\fleq_{\scriptscriptstyle L'U'}}
\newcommand{\lu}{\fleq_{\scriptscriptstyle LU}}

\renewcommand{\int}[1]{\lfloor #1 \rfloor}

\newcommand{\xra}{\xrightarrow}

\newcommand{\Extra}{\mathit{Extra}}

\newcommand{\abstr}{\mathfrak{a}}
\newcommand{\abs}{\abstr}
\newcommand{\tto}{\Rightarrow}

\newcommand{\tuple}[1]{\ensuremath{(#1)}}
\newcommand{\Rpos}[0]{\ensuremath{\mathbb{R}_{\geq 0}}}
\newcommand{\Zall}[0]{\ensuremath{\mathbb{Z}}}
\renewcommand{\Nat}[0]{\ensuremath{\mathbb{N}}}
\renewcommand{\CC}[0]{\ensuremath{\Phi}} 
\newcommand{\val}[0]{\ensuremath{v}} 
\newcommand{\vali}[0]{\ensuremath{\mathbf{0}}} 

\newcommand{\elapse}[1]{\ensuremath{\overrightarrow{#1}}}

\newcommand{\reset}[1]{\ensuremath{[#1]}}

\newcommand{\ExtraLUp}[0]{\ensuremath{{\Extra_{LU}^+}}}

\newcommand{\tasim}{\preceq_{t.a.}}
\newcommand{\abssim}{\abs_{\tasim}}

\newcommand{\lleq}{\lessdot}
\newcommand{\ggeq}{\gtrdot}

\newcommand{\Post}{\operatorname{Post}}

\begin{document}

\title{Lazy abstractions for timed automata}

\author{F. Herbreteau\inst{1}, 
  B. Srivathsan\inst{2}, and I. Walukiewicz\inst{1}}
\institute{
  Univ. Bordeaux, CNRS, LaBRI, UMR 5800, F-33400 Talence, France
  \and
  Software Modeling and Verification group, RWTH Aachen University,
  Germany
}
\maketitle

\begin{abstract}
  We consider the reachability problem for timed automata.
  A standard solution to this problem involves computing a search tree
  whose nodes are abstractions of zones. 
  For efficiency reasons, they are parametrized by the maximal lower
  and upper bounds ($LU$-bounds) occurring in the guards of the
  automaton. 
  We propose an algorithm that is updating $LU$-bounds during exploration of
  the search tree.
  In order to keep them as small as possible,
  the bounds are refined only when they enable a transition
  that is impossible in the unabstracted system.
  So our algorithm can be seen as a kind of lazy CEGAR algorithm for
  timed automata. 
  We show that on several standard benchmarks, the algorithm
  is capable of keeping very small $LU$-bounds, and in
  consequence reduce the search space substantially. 
\end{abstract}

\section{Introduction}

Timed automata are obtained from finite automata by adding clocks that
can be reset and whose values can be compared with constants. 
The reachability problem asks if a given target state is reachable
from the initial state by the execution of a given automaton.
The standard solution to this problem involves computing,
so called, zone graph of the automaton, and the use of abstractions to
make the algorithm both terminating and more efficient.

Most abstractions are based on constants used in comparisons of clock values.
Such abstractions have already been considered in the seminal paper of Alur and
Dill~\cite{Alur:TCS:1994}. Behrmann et. al.~\cite{Behrmann:STTT:2006}
have proposed abstractions based on so called
$LU$-bounds, that are functions giving for every clock a maximal
constant in a lower, respectively upper bound, constraint in the
automaton.
In a recent paper~\cite{lics} we have shown how to efficiently use $\alu$
abstraction from~\cite{Behrmann:STTT:2006}.  Moreover, $\alu$ has been proved to be the
biggest abstraction that is sound for all automata with given
$LU$-bounds.  Since $\alu$ abstraction of a zone can result in a
non-convex set, we have shown in op.\ cit.\ how to use this
abstraction without the need to store the result of the abstraction.
This opens new algorithmic possibilities because  changing $LU$-bounds
becomes very cheap as abstractions need not be recalculated.  In
this paper we explore these possibilities.

The algorithm we propose works as follows. 
It constructs a graph with nodes of the form $(q,Z,LU)$, where $q$
is a state of the automaton, $Z$ is a zone, and $LU$ are parameters for
the abstraction.
It starts with the biggest abstraction: $LU$ bounds are set to
$-\infty$ which makes $\alu(Z)$ to be the set of all valuations for
every nonempty $Z$. The algorithm explores the zone graph using
standard transition relation on zones, without modifying $LU$ bounds
till it encounters a disabled transition. 
More concretely, till it reaches a node $(q,Z,LU)$ such that there is a
transition from $q$ that is not possible from $(q,Z)$ because no
valuation in $Z$ allows to take it. At this point we need to adjust
$LU$ bounds so that the transition is not possible from $\alu(Z)$
either. 
This adjustment is then propagated backwards through already
constructed part of the graph. 

The real challenge is to limit the propagation of bound updates.
For this, if the bounds have changed in a node $(q',Z',L'U')$ then we
consider its predecessor nodes $(q,Z,LU)$ and update its $LU$ bounds
as a function of $Z$, $Z'$ and $L'U'$.  
We give general conditions for
correctness of such an update, and a concrete efficient algorithm
implementing it. 
This requires getting into a careful analysis of the influence of the
transition on the zone $Z$. 
In the result we obtain an algorithm that exhibits exponential gains
on some standard benchmarks.

We have analyzed the performance of our algorithm theoretically as
well as empirically. 
We have compared it with static analysis algorithm that is the
state-of-the-art algorithm implemented in UPPAAL, and with an algorithm
we have proposed in~\cite{Herbreteau:FSTTCS:2011}. 
The later improves on the static analysis algorithm by considering
only the reachable part of the zone graph. 
For an example borrowed from~\cite{LugiezNiebertZ05} we have proved that the
algorithm presented here produces a linear size search graph while for
the other two algorithms, the search graph is exponential in the
size of the model.
For the classic FDDI benchmark, that has been tested on just about
every algorithm for the reachability problem, our algorithm shows
rather surprising fact that the time is almost irrelevant. 
There is only one constraint that induces $LU$ bounds, and in
consequence the abstract search graph constructed by our algorithm is
linear in the size of the parameter of FDDI.

Our algorithm can be seen as a kind of CEGAR algorithm similar in the
spirit to~\cite{HenJMSutre02}, but then there are also major differences.
In the particular setting of timed automata the information available
is much richer, and we need to use it in order to obtain a competitive
algorithm. 
First, we do not need to wait till a whole path is constructed to
analyze if it is spurious or not. 
Once we decide to keep zones in nodes we can immediately detect if an
abstraction is too large: it is when it permits a transition not
permitted from the zone itself. 
Next, the abstractions we use are highly specialized for
the reachability problem. Finally, the propagation of bound changes
gets quite sophisticated because it can profit from the large
amount of useful information in the exploration graph.

\paragraph{Related work}
Forward analysis is the main approach for the reachability testing of
real-time systems. The use of zone-based abstractions for termination
has been introduced in~\cite{Daws:TACAS:1998}.  The notion of
$LU$-bounds and inference of these bounds by static analysis of an
automaton have been proposed
in~\cite{Behrmann:TACAS:2003,Behrmann:STTT:2006}. The $\alu$
approximation has been introduced in~\cite{Behrmann:STTT:2006}. An
approximation method based on LU-bounds, called $\Extra^+_{LU}$, is
used in the current implementation of UPPAAL~\cite{Behrmann:QEST:2006}.
In~\cite{lics} we have shown how to efficiently use $\alu$
approximation. We have also proposed an $LU$-propagation
algorithm~\cite{Herbreteau:FSTTCS:2011} that
can be seen as applying the static analysis
from~\cite{Behrmann:TACAS:2003} on the zone graph instead of the graph
of the automaton; moreover this inference is done on-the-fly
during construction of the zone graph. In the present paper we do much
finer inference and propagation of $LU$-bounds.

Approximation schemes for analysis of timed-automata have been
considered almost immediately after introduction of the concept of
timed automata, as for example in~\cite{Wong-Toi-PhD,DillW95}
or~\cite{Sorea04}. In particular, the later citation  proposes to
abstract the region graph by not considering all the
constraints involved in the definition of a region. When a spurious
counterexample is discovered a new constraint is added. So in the
worst case the whole region graph will be constructed.  Our algorithm
in the worst case constructs an $\alu$-abstracted zone graph with $LU$-bounds obtained by static analysis.  This is as good as
state-of-the-art method used in UPPAAL. Another slightly related paper
is~\cite{BouyerLR05} where CEGAR approach is used to handle diagonal
constraints.

Let us mention that abstractions are not needed in backward
exploration of timed systems. Nevertheless, any feasible backward
analysis approach needs to simplify constraints. For
example~\cite{Morbe:CAV:2011} does not use approximations and relies
on an SMT solver instead. This approach, or the approach of
RED~\cite{WangRED04}, are very difficult to 
compare with the forward analysis approach we study here.

\paragraph{Organization of the paper} In the preliminaries section we
introduce all standard notions we will need, and $\alu$ abstraction in
particular. 
Section~\ref{sec:asg} gives a definition of adaptive simulation
graph (ASG). Such a graph represents the search space of a forward
reachability testing algorithm that will search for an abstract run
with respect to $\alu$ abstraction, while changing $LU$-bounds
dynamically during exploration. 
Section~\ref{sec:algo} gives an algorithm for constructing an ASG with
small $LU$-bounds. 
Section~\ref{sec:newbounds} presents the two crucial functions used in
the algorithm: the one updating the bounds due to disabled edges, and
the one propagating the change of bounds. Section~\ref{sec:example}
explains some advantages of algorithm on variations of an example
borrowed from~\cite{LugiezNiebertZ05}. The experiments section compares our
prototype tool with UPPAAL, and our algorithm
from~\cite{Herbreteau:FSTTCS:2011}. Conclusions section gives  some justification for our
choice of concentrating on $LU$-bounds.



\section{Preliminaries}


\subsection{Timed automata and the reachability problem}
\label{sec:timed-automata}

Let $X$ be a set of clocks, i.e., variables that range over $\Rpos$,
the set of non-negative real numbers. A \emph{clock constraint} is a
conjunction of constraints $x\# c$ for $x\in X$,
$\#\in\{<,\leq,=,\geq,>\}$ and $c\in \Nat$, e.g. $(x \le 3 \wedge y >
0)$. Let $\CC(X)$ denote the set of clock constraints over clock
variables $X$.  A \emph{clock valuation} over $X$ is a function
$\val\,:\,X\rightarrow\Rpos$. We denote $\Rpos^X$ the set of clock
valuations over $X$, and $\vali$ the valuation that associates $0$ to
every clock in $X$. We write $\val\sat \phi$ when $\val$ satisfies
$\phi\in \CC(X)$, i.e. when every constraint in $\phi$ holds after
replacing every $x$ by $\val(x)$. For $\d\in\Rpos$, let $\val+\d$ be
the valuation that associates $\val(x)+\d$ to every clock $x$. For
$R\subseteq X$, let $\reset{R}\val$ be the valuation that sets $x$ to
$0$ if $x\in R$, and that sets $x$ to $\val(x)$ otherwise.

A \emph{timed automaton (TA)} is a tuple $\Aa=\tuple{Q,q_0,X,T,\Acc}$
where $Q$ is a finite set of states, $q_0\in Q$ is the initial state,
$X$ is a finite set of clocks, $\Acc\subseteq Q$ is a set of accepting
states, and $T\,\subseteq\, Q\times \CC(X)\times 2^X \times Q$ is a
finite set of transitions $\tuple{q,g,R,q'}$ where $g$ is a
\emph{guard}, and $R$ is the set of clocks that are \emph{reset} on
the transition.

A \emph{configuration} of $\Aa$ is a pair $(q,v)\in Q\times\Rpos^X$
and $(q_0,\vali)$ is the \emph{initial configuration}. We have two
kinds of transitions:

\smallskip

\par\noindent\textbf{Delay:} $(q,v)\to^\d(q,v+\d)$ for some $\d\in
\Rpos$;

\smallskip

\par\noindent\textbf{Action:} $(q,v)\to^t(q,v')$ for some transition
$t=(q,g,R,q')\in T$ such that $v\sat g$ and $v'=[R]v$.

A \emph{run} of $\Aa$ is a finite sequence of transitions starting
from the initial configuration $(q_0, \vali)$. Without loss of
generality, we can assume that the first transition is a delay
transition and that delay and action transitions alternate. We write
$(q, v) \xra{\d, t} (q',v')$ if there is a delay transition $(q,v)
\to^\d (q, v+\d)$ followed by an action transition $(q, v+\d) \to^{t}
(q',v')$.  So a run of $\Aa$ can be written as:
\begin{equation*}
  (q_0, v_0) \xra{\d_0, t_0} 
  (q_1, v_1) \xra{\d_1, t_1} 
  (q_2, v_2) \dots 
  (q_n, v_n)
\end{equation*}
where $(q_0, v_0)$ represents the initial configuration $(q_0,\vali)$.
 
A run is \emph{accepting} if it ends in a configuration
$\tuple{q_n,v_n}$ with $q_n\in \Acc$. 

\begin{definition}[Reachability problem]
  The \emph{reachability problem} for timed automata is to decide
  whether there exists an accepting run of a given automaton.
\end{definition}
This problem is known to
be~\PSPACE-complete~\cite{Alur:TCS:1994,Courcoubetis:FMSD:1992}.
The class of TA we consider is usually known as diagonal-free TA since
clock comparisons like $x-y\leq 1$ are disallowed. Notice that if we
are interested in state reachability, considering timed automata
without state invariants does not entail any loss of generality as the
invariants can be added to the guards. For state reachability, we can
also consider automata without transition labels.

\subsection{Zones and symbolic runs}

Here we introduce zones that are sets of valuations defined by simple
linear constraints. We also define symbolic transition relation working
on sets of valuations. These definitions will allow us to concentrate
on symbolic runs instead of concrete runs as in the previous section. 

We first define a transition relation $\tto$ over nodes of the form
$(q,W)$ where $W$ is a set of valuations.

\begin{definition}[Symbolic transition
  $\tto$]\label{def:transition-relation}
  Let $\Aa$ be a timed automaton. For every transition $t$ of $\Aa$
  and every set of valuations $W$, we have a transition $\tto^t$
  defined as follows:
  \begin{align*}
    & (q, W) \tto^t (q',W')  \text{ where } W' = \{ v' ~|~ \exists
    v \in W, ~\exists \d \in \Rpos.~ (q,v) \to^t \to^\d (q',v') \}
  \end{align*}
  We will sometimes write $\Post_t(W)$ for $W'$. 
  The transition relation $\tto$ is the union of all $\tto^t$.
\end{definition}

\medskip

The transition relation defined above considers each valuation $v \in
W$ that can take the transition $t$, obtains the valuation after the
transition and then collects the
time-successors from this obtained valuation. Therefore the symbolic
transition $\tto$ always yields sets closed under time-successors. 
The initial configuration of the automaton is $(q_0, \vali)$. Starting
from the initial valuation $\vali$ the set of valuations reachable by a
time elapse at the initial state are given by $\{ \vali + \d ~|~ \d
\in \Rpos\}$. Call this $W_0$. From $(q_0, W_0)$ as the initial node,
computing the symbolic transition relation $\tto$ leads to different nodes $(q,
W)$ wherein the sets $W$ are closed under time-successors.


It has been noticed in~\cite{Bengtsson:Springer:2004} that the sets $W$ obtained in the
nodes $(q,W)$ can be described by some simple constraints involving
only the difference between clocks.
This has motivated the definition of \emph{zones}, which are sets of
valuations defined by difference constraints.

\begin{definition}[Zones~\cite{Bengtsson:Springer:2004}]
  A zone is a set of valuations defined by a conjunction of two kinds
  of clock constraints: $x \sim c$ and $x - y \sim c$
  for $x, y \in X$, $\sim \in \{ \le, <, = , >, \ge\}$, and $c \in
  \Zall$. 
\end{definition}

For example $(x>4 \land y-x\leq 1)$ is a zone.
It can be shown that starting from a node $(q,W)$ with $W$ being a
zone, the transition $(q,W) \tto (q',W')$ leads to a node in which
$W'$ is again a zone~\cite{Bengtsson:Springer:2004}.
 Observe that the initial set of valuations $Z_0 = \{ \vali + \d~|~\d
 \in \Rpos \}$
is indeed a zone: it is given by the constraints 
$\Land_{x,y \in X} ~(x \ge 0 ~\land~ x - y =0)$

These observations lead to a notion of \emph{symbolic run} that is a
sequence of symbolic transitions
\begin{equation*}
  (q_0,Z_0)\tto (q_1,Z_1)\tto\dots
\end{equation*}
\begin{proposition}
  Fix a timed automaton.
  The automaton has an accepting run if and only if there it has a symbolic run
  reaching an accepting state and non-empty zone.
\end{proposition}
This proposition does not yet give a complete solution to the
reachability problem since there may be infinitely many reachable
zones, so it is not immediate how to algorithmically check that a symbolic
run does not exist. A standard solution to this problem of
non-termination is to use abstractions that we introduce in the next
subsection.



\subsection{Bounds and abstractions}

In the previous subsection, we have defined zones. We have used zones
instead of valuations to solve the reachability problem.  Since the
number of reachable zones can be infinite, the next step is to group
zones together into a finite number of equivalence classes.  An
\emph{abstraction operator} is a convenient way to express a grouping
of valuations, and in consequence grouping of zones. Instead of
discussing abstractions in full generality, we will immediately
proceed to the most relevant case of abstractions based on
time-abstract simulation~\cite{Tasiran:CONCUR:1996}.

For this subsection we fix a timed automaton $\Aa$.

\begin{definition}[Time-abstract simulation]
  \label{def:time-abstr-simulation}
  A \emph{(state based) time-abstract simulation} between
  configurations of $\Aa$ is a relation $(q, v) \tasim (q',v')$
  such that:
  \begin{itemize}
  \item $q = q'$,
  \item if $(q,v) \to^\d (q, v + \d) \to^t (q_1, v_1)$, then there
    exists a $\d' \in \Rpos$ such that $(q,v') \to^{\d'} (q, v' + \d')
    \to^t (q_1, v_1')$ satisfying $(q_1, v_1) \tasim (q_1,v_1')$ for
    the same transition $t$.
  \end{itemize}
\end{definition}

For two valuations $v,v'$, we say that $v \tasim v'$ if for every
state $q$ of the automaton, we have $(q,v) \tasim (q',v')$.  An
abstraction $\abssim$ based on a simulation $\tasim$ can be defined as
follows:

\begin{definition}[Abstraction based on simulation]
  \label{def:abstr-based-ta-sim}
  Given a set $W$, we define $\abssim(W) = \{ v~|~ \exists v' \in W.\
  v \tasim v'\}$.
  The abstract transition relation is $(q,W)\tto_{\abssim}(q',\abssim(W'))$
  where $W=\abssim(W)$ and $(q,W)\tto (q',W')$ (cf.\
  Definition~\ref{def:transition-relation}).  
\end{definition}

Let $\tto_{\abssim}^*$ denote the reflexive and transitive closure of
$\tto_{\abssim}$. Similarly, let $\to^*$ denote the reflexive and
transitive closure of the transition relation $\to$ of the automaton. 
It can be easily verified that the
abstract transition relation satisfies the following two important
properties ($W_0$ denotes $\{\vali + \d~|~ \d \in \Rpos\}$)

\begin{description}
\item[Soundness:] if $(q_0,W_0)\tto_{\abssim}^*(q,W)$ then there is
  $v\in W$ such that $(q_0,\vali)\to^*(q,v)$.

\item[Completeness:] if $(q_0,\vali)\to^*(q,v)$ then there is $W$ such
  that $v\in W$ and $(q_0,W_0)\tto_{\abssim}^*(q,W)$.
\end{description}
These properties immediately imply that abstract transitions can be
used to solve the reachability problem.

\begin{proposition}\label{prop:abstract run}
  For every abstraction operator $\abssim$ based on timed-abstract simulation.
  Automaton $\Aa$ has a run reaching a state $q$ iff there is an
  abstract run
  \begin{equation*}
    (q^0,W_0)\tto_{\abssim} (q_1,W_1)\tto_{\abssim}\dots\tto_{\abssim}(q,W)
  \end{equation*}
  for some $W\not=\es$. 
\end{proposition}

\begin{remark}
  \label{rem:coarse-abstractions}
  If $\abstr$ and $\mathfrak{b}$ are two abstractions such that for
  every set of valuations $W$ we have $\abstr(W) \incl
  \mathfrak{b}(W)$ then we prefer to use $\mathfrak{b}$ since every
  abstract run with respect to $\mathfrak{a}$ is also a run with
  respect to $\mathfrak{b}$. In consequence, it is easier to find an
  abstract run for $\mathfrak{b}$ abstraction. 
\end{remark}

Therefore, the aim is to come up with a finite abstraction as coarse
as possible, that still maintains the soundness property.

For a given automaton it can be computed if two
configurations are in a simulation relation. It should be noted though
that computing the coarsest simulation relation is
\EXPTIME-hard~\cite{larsch00}. Since the reachability problem can be
solved in \PSPACE{}, this suggests that it may not be reasonable to
try to solve it using the abstraction based on the coarsest
simulation. We can get simulation relations that are computationally
easier if we consider only a part of the structure of the
automaton. The common way is to look at constants appearing in the
guards of the automaton and consider them as parameters for
abstraction.

\subsection{LU-bounds and LU-abstractions}

The most common parameter taken for defining
abstractions are $LU$-bounds. 

\begin{definition}[LU-bounds]\label{def:bounds}
  The $L$ bound for an automaton $\Aa$ is the function assigning to
  every clock $x$ a maximal constant that appears in a lower bound
  guard for $x$ in $\Aa$, that is, maximum over guards of the form $x
  > c$ or $x \ge c$. Similarly $U$ is the function assigning to every
  clock $x$ a maximal constant appearing in an upper bound guard for
  $x$ in $\Aa$, that is, maximum over guards of the form $x < c$ or $x
  \le c$.
\end{definition}

The paper introducing LU-bounds \cite{Behrmann:STTT:2006} also
introduced an abstraction operator $\alu$ that uses LU-bounds as parameters.
We begin by recalling the definition of an LU-preorder defined
in~\cite{Behrmann:STTT:2006}. We use a different but equivalent
formulation.

\begin{definition}[LU-preorder~\cite{Behrmann:STTT:2006}]
  \label{def:lu-preorder}
  Let $L,U:X\to\Nat \cup \{-\infty\}$ be two bound functions. For a
  pair of valuations we set $v\lu v'$ if for every clock $x$:
  \begin{itemize}
  \item if $v'(x)<v(x)$ then $v'(x)> L_x$, and
  \item if $v'(x)>v(x)$ then $v(x)>U_x$.
  \end{itemize}
\end{definition}

It has been shown in ~\cite{Behrmann:STTT:2006} that $\lu$ is a
time-abstract simulation relation.  The $\alu$ abstraction is based on
this LU-preorder $\lu$.

\begin{definition}[$\alu$-abstraction~\cite{Behrmann:STTT:2006}]
  Given $L$ and $U$ bound functions, for a set of valuations $W$ we
  define:
  \begin{equation*}
    \alu(W)=\set{v ~|~ \exists v'\in W.\  v\lu v'}.
  \end{equation*}
\end{definition}

\begin{figure}[!t]
  \centering \begin{tikzpicture}

    
  \begin{scope}[xshift = -1cm]
    \draw[thick] (2,0) -- (2,4); \draw (2, -0.3) node {\scriptsize
      $\mathbf{U_x}$} (2.5, -0.4); \draw[thick] (3.5,0) --
    (3.5,4); \draw (3.5, -0.3) node {\scriptsize
     $\mathbf{L_x}$} (4, -0.4);
  \end{scope}
  
  \begin{scope}[yshift=-0.5cm]
    \draw[thick] (0,1.5) -- (4,1.5); \draw (-0.4, 1.5) node
    {\scriptsize $\mathbf{L_y}$} (0,1.5); \draw[thick]
    (0,3) -- (4,3); \draw (-0.4, 3) node {\scriptsize
      $\mathbf{U_y}$} (0, 3);
  \end{scope}

  \draw[->, very thick, >=stealth] (0,0) -- (0,4); \draw[->, very
  thick, >=stealth] (0,0) -- (4,0); \draw
  (-0.2,-0.2) node {\scriptsize $0$} (0, -0.4); \draw (4.2, -0.2) node
  {\scriptsize $x$} (4.2, -0.2); \draw (-0.2, 4.2) node {\scriptsize
    $y$} (0,4.2);

  \draw[thick, fill=gray, nearly transparent] (0.5,0)
    -- (4,3.5) -- (4,0) -- cycle;
  \draw (3.25, 1.75) node {\scriptsize $Z$ } (3.5, 1.75); 

 \fill[pattern = dots] (1,0.5) -- (1,4) -- (4,4) -- (4,3.5) -- cycle; 



 \begin{scope}[xshift = -3cm, yshift=-1cm]
    \node[left] at (9.5, 3.1) {\footnotesize $Z :$};
    \node[left] at (9.5, 2.5) {\footnotesize $\alu(Z):$};

    \draw[thick, fill=gray, nearly transparent] (9.5,2.9) rectangle
   (9.9, 3.3);
    \draw[thick, fill=gray, nearly transparent] (9.5,2.3) rectangle
  (9.9, 2.7);
   
   \node at (10.2, 2.5) {$\cup$};
   \draw[fill, pattern = dots] (10.5,2.3) rectangle
   (10.9, 2.7);
\end{scope}

\end{tikzpicture}

  \caption{Zone $Z$ is given by the grey area. Abstraction $\alu(Z)$
    is given by the grey area along with the dotted area}
  \label{fig:alu-example}
\end{figure}
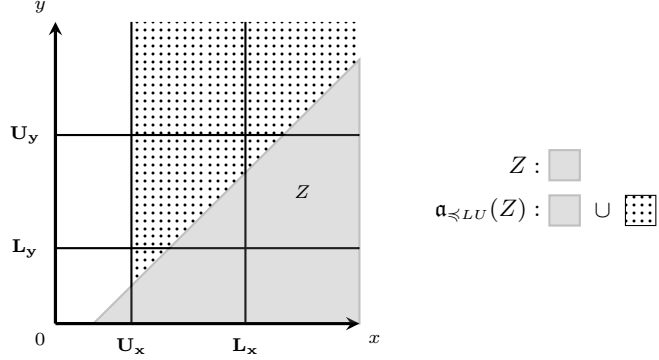

Figure~\ref{fig:alu-example} gives an example of a zone $Z$ and its
abstraction $\alu(Z)$. It can be seen that $\alu(Z)$ is not a convex
set.

An efficient algorithm to use the $\alu$ abstraction for reachability
was proposed in~\cite{lics}. Moreover in op cit. it was shown that
over time-elapsed zones, $\alu$ abstraction is optimal when the only
information about the analyzed automaton are its $LU$-bounds. Informally
speaking, for a fixed $LU$, the $\alu$ abstraction is the biggest
abstraction that is sound and complete for all automata using guards
within $LU$-bounds.

Since the abstraction $\alu$ is optimal, the next improvement is to
try to get as good $LU$-bounds as possible since tighter bounds give
coarser abstractions.  Recall Remark~\ref{rem:coarse-abstractions}
which states the importance of having coarser abstractions.  

It has been proposed in~\cite{Behrmann:TACAS:2003} that instead of
considering one $LU$-bound for all states in an automaton, one can use
different bound functions for each state. For every
state $q$ and every clock $x$, constants $L_x(q)$ and $U_x(q)$ 
are determined by the least solution of the following set of
inequalities. For each transition $(q, g, R, q')$ in the automaton, we
have:
\begin{align}\label{eq:static}
  \begin{cases}
    L_x(q) \ge c & \text{ if } x \ggeq c \text{ is a constraint in } g \\
    L_x(q) \ge L_x(q') & \text{ if } x \not \in R
  \end{cases}
\end{align}
Similar inequalities are written for $U$, now considering $x \lleq
c$. It has been shown in~\cite{Behrmann:TACAS:2003} that such an
assignment of constants is sound and complete for state reachability.
Experimental results have shown that this method, that performs a
static analysis on the structure of the automaton, often gives
very big gains.



\section{Adaptive simulation graph}\label{sec:asg}

In this paper we improve on the idea of static analysis that computes
$LU$-bounds for each state $q$. We will compute $LU$-bounds on-the-fly
while  searching for an abstract run. 
The immediate gain will be that bounds will depend not only on a state
but also on a set of valuations.  
The real freedom given by an adaptive simulation graph and
Theorem~\ref{thm:algo-correct} presented below is that they
will allow to ignore some guards of transitions when calculating the
$LU$ bounds. 
As we will see in experimental section, this can result in very big
performance gains.

We will construct forward reachability testing algorithm that will
search for an abstract run with respect to $\alu$ abstraction, where
$LU$ bounds will change dynamically during exploration.  
The intuition of a search space of such an algorithm is formalized in
a notion of adaptive simulation graph (ASG). 
Such a graph permits to change $LU$ bounds from node to node,
provided some consistency conditions are satisfied.
$LU$-bounds play an important role in this graph.
They are used to stop developing successors of a node as soon as
possible. 
So our goal will be to find as small $LU$-bounds as possible in order
to cut the paths of the graph as soon as possible.

\begin{definition}[Adaptive simulation graph (ASG)]\label{df:asg}
Fix an automaton $\Aa$. An ASG graph has nodes of the form 
 $(q,Z,LU)$ where $q$ is the state of
$\Aa$, $Z$ is a zone, and $LU$ are bound functions. Some nodes are
declared to be tentative. The graph
is required to satisfy three conditions:
\begin{description}
\item[G1] For the initial state $q^0$ and initial zone $Z_0$, a node
  $(q_0,Z_0,LU)$ should appear in the graph for some $LU$.
\item[G2] If a node $(q,Z,LU)$ is not tentative then for every
  transition $(q,Z)\tto_t (q',Z')$ the node
  should have a successor labeled $(q',Z',L'U')$ for some $L'U'$.
\item[G3] If a node $(q,Z,LU)$ is tentative then there should be
  non-tentative node $(q',Z',L'U')$ such that $q=q'$ and $Z\incl \aluP(Z')$. Node
  $(q,Z',L'U')$ is called covering node. 
\end{description}
We will also require that the following invariants are satisfied:
\begin{description}
\item[I1] If a transition $\tto_t$ is disabled from $(q,Z)$, and
  $(q,Z,LU)$ is a node of the ASG then $\tto_t$ should be disabled from
  $\alu(Z)$ too;
\item[I2] For every edge $(q, Z, LU) \tto_t (q',Z',L'U')$ the ASG we have:
  \begin{align*}
    \post_t(\alu(Z)) \incl \aluP(Z').
  \end{align*}
\item[I3] For every tentative node $(q, Z_1, L_1U_1)$ and the
  corresponding covering node $(q, Z_2, L_2U_2)$, we have:
  \begin{align*}
    L_2U_2 \le L_1U_1.
  \end{align*}
\end{description}
\end{definition}
The conditions G1, G2, G3 express the expected requirements for a
graph to cover all reachable configurations. 
In particular, the condition G3 allows to stop exploration if there is
already a ``better'' node in the graph. 
The three invariants are more subtle. They imply that $LU$-bounds
should be big enough for the reachability information to be preserved.
(cf.~Theorem~\ref{thm:algo-correct}).

\textbf{Remark:} While the idea is to work with nodes of the form
$(q,W)$ with $W=\alu(W)$, we do not want to store $W$ directly, as we
have no efficient way of representing and manipulating such
sets. Instead we represent each $W$ as $\alu(Z)$. So we store $Z$ and
$LU$. This choice is algorithmically cheap since testing the
inclusion $Z'\incl \alu(Z)$ is practically as easy as testing $Z'\incl
Z$~\cite{lics}. This approach has another big advantage: when we change $LU$
bound in a node, we do not need to recalculate $\alu(Z)$.

\textbf{Remark:}\label{rem:asg-finite} It is important to observe that for every $\Aa$ there
exists a finite ASG. 
For example, it is sufficient to take static $LU$-bounds as described
in~\eqref{eq:static}. 
It means that we can take ASG whose nodes are $(q,Z,L(q)U(q))$ with
bound functions given by static analysis. 
It is easy to see that such a choice makes all three invariants hold.

The next theorem tells us that any ASG is good enough to determine the
existence of an accepting run. Our objective in the later section will
be to construct as small ASG as possible.

\begin{theorem}\label{thm:algo-correct}
  Let $G$ be an ASG for an automaton $\Aa$. An accepting state is reachable 
  by a run of $\Aa$ iff a node containing an accepting state of $\Aa$ and a
  non-empty zone is reachable from the initial node of $G$.
\end{theorem}

Recall from Proposition~\ref{prop:abstract run} that there is an
accepting run of $\Aa$ iff there 
is a sequence of symbolic transitions
\begin{equation}\label{eqn:zone path}
  (q_0,Z_0)\tto (q_1,Z_1)\tto\dots\tto (q,Z)
\end{equation}
with $q\in \Acc$ and $Z\not=\es$.

For the right-to-left direction of the theorem we take a path in $G$ 
leading from $(q_0,Z_0,L_0U_0)$ to $(q,Z,LU)$. By definition, removing
the third component gives us a path as in~\eqref{eqn:zone path}.

The opposite direction is proved with the help of the following lemma.

\begin{lemma}
  Let $(q,Z)$ be as in~\eqref{eqn:zone path}. There exists a non
  tentative node $(q,Z_{1},L_1U_1)$ in $G$ such that $Z \subseteq
  \aluOne(Z_1)$.
\end{lemma}
\begin{proof}
  The lemma is vacuously true for $(q_0, Z_0)$. Assume that the
  hypothesis is true for a path as in~\eqref{eqn:zone path}.  We prove
  that the lemma is true for every symbolic successor of $(q,Z)$.

  Let $(q,Z) \tto^t (q', Z')$ be a symbolic transition of $\Aa$. The
  transition $\tto^t$ should be enabled from $(q, Z_1)$. This
  is because if it was disabled, by Invariant 1, we would have that it
  is disabled from $\aluOne(Z_1)$ and from the hypothesis, it should
  be disabled from $(q,Z)$ too leading to a contradiction.

  So we have a transition $(q, Z_1, L_1U_1) \tto^t (q', Z_1',
  L'_1U'_1)$ in $G$. From Invariant 2, we have
  $\post(\aluOne(Z_1)) \incl \aluOneP(Z'_1)$. This leads to the
  following sequence of implications.

  \begin{align*}
    & & Z & ~\incl~ \aluOne(Z_1) \qquad \quad \text{ induction hypothesis} \\
    & \imp & \post(Z) & ~\incl~ \post(\aluOne(Z_1)) \\
    & \imp & \post(Z) & ~\incl~ \aluOneP(Z'_1) \qquad \quad \text{ by
      Invariant 2}\\
    & \imp & Z' & ~\incl~ \aluOneP(Z'_1)
  \end{align*}

  If $(q', Z'_1, L'_1U'_1)$ is a non-tentative node, then we are
  done. Suppose it is a tentative node, then we know that there exists
  a non-tentative node $(q', Z'_2, L'_2U'_2)$ such that $Z_1' \incl
  \aluTwoP(Z_2')$. From Invariant 3, we also know that $L_2'U_2' \le
  L_1'U_1'$. This shows that $Z' \incl \aluTwoP(Z'_2)$.

  Hence the node corresponding to $(q',Z')$ is $(q', Z_2',
  L_2'U_2')$.\qed
\end{proof}



\section{Algorithm}\label{sec:algo}

Our aim is to construct a small adaptive simulation graph for a given
timed automaton.
For this the algorithm will try to keep $LU$ bounds as small as
possible but still satisfy the invariants I1, I2, I3. 
The bounds are calculated dynamically while constructing an adaptive
simulation graph. 
For example, the invariant I1 requires $LU$ in a node to be sufficiently
big so that the transition remains disabled. Invariant I2 tells that
$LU$ bound in a node should depend on $LU$ bounds in the successors of
the node. 

\emph{Proviso:}\label{proviso} For simplicity of the algorithm
presented in this section we assume a special form of transitions of
timed automata. A transition can have either only upper bound guards,
or only lower bound guards and no resets. Observe that a transition
$q_1\xra{g;R} q_2$ is equivalent to $q_1\xra{g_L}
q'_1 \xra{g_U;R} q_2$; where $g_L$ is the conjunction of the lower
bound guards from $g$ and $g_U$ is the conjunction of the upper bound
guards from $g$.  
\begin{lemma}
  Suppose $W_1$ is a time elapsed set of valuations. If
  \begin{equation*}
    (q_1,W_1)\tto_{g;R}(q_2,W_2)\quad \text{and}\quad
    (q_1,W_1)\tto_{g_L}(q'_1,W'_1)\tto_{g_U;R}(q_2,W'_2)
  \end{equation*}
then $W_2=W'_2$.
\end{lemma}
\begin{proof}
  We consider only, more complicated, inclusion $W'_2\incl W_2$.
  Take $v'_2\in W'_2$. By definition we know that there is $v_1\in
  W_1$ such that 
  \begin{equation*}
    (q_1,v_1)\xra{g_L}(q'_1,v_1+\d_1)\xra{g_U;R}(q_2,(v_1+\d_1)[R]+\d_2)
  \end{equation*}
  and $v_2=(v_1+\d_1)[R]+\d_2$. We get then
  \begin{equation*}
    (q_1,v_1+\d_1)\xra{g_L}(q'_1,v_1+\d_1)\xra{g_U;R}(q_2,(v_1+\d_1)[R]+\d_2)
  \end{equation*}
  So $(q_1,v_1+\d_1)\xra{g;R}(q_2,(v_1+\d_1)[R]+\d_2)$. As $W_1$
  is time elapsed, $v_1+\d_1\in W_1$. This shows $v_2\in W_2$, by
  definition of $W_2$. \qed
\end{proof}
 So in order to satisfy our proviso we may need to
double the number of states of an automaton.

{
  \setlength{\columnsep}{35pt}
  \setlength{\columnseprule}{0.5pt}
\begin{lstlisting}[mathescape=true, float=tp, caption={ Reachability
    algorithm with on-the-fly bound computation and $\alu$ abstraction.},
label=fig:algorithm]
function main():
    let $v_{root}$ be the root node with $v_{root}.q = q_0$ and $v_{root}.Z = Z_0$
    add $v_{root}$ to the stack
    while (stack $\neq$ $\es$) do
        remove $v$ from the stack
        explore($v$)
        resolve()
    return "empty"

procedure explore($v$):
    if ($v.q$ is accepting)
        exit "not empty"
    if ($\exists$ $v''$ nontentative s.t. $v.q=v''.q$ and $v.Z\subseteq\ALU{v''.LU}(v''.Z)$)
        mark $v$ tentative wrt $v''$
        $v.LU$ := $v''.LU$
        $(X_L,X_U)$:= active clocks in $v.LU$
        propagate($v$,$X_L$,$X_U$)
    else
        $v.LU$ := disabled($v.q$,$v.Z$)
        $(X_L,X_U)$:= active clocks in $v.LU$
        propagate($v$,$X_L$,$X_U$) 
        for each $(q',Z')$ s.t. $(v.q, v.Z) \tto (q',Z')$ and $Z'\not=\es$ do
          create $v'$ the successor of $v$ with $v'.q = q'$ and $v'.Z = Z'$
          explore($v'$)

function disabled($q$,$Z$)
    examine transitions from $q$ that are disabled from $Z$ and
    choose $LU$ so that invariant $I1$ is satisfied
    return($LU$);
  
procedure resolve():
    for each $v$ tentative w.r.t. $v'$ do
        if $v.Z \not\incl\ALU{v'.LU}(v'.Z)$
            mark $v$ nontentative
            set $v.L$ and $v.U$ to $-\infty$ // clear the bounds in $v$
            add $v$ to stack

procedure propagate($v'$,$X'_L$,$X'_U$):
    $v$=parent($v'$);
    $LU$ := newbounds($v$,$v'$,$X'_L$,$X'_U$)
    if ($LU\not= v.LU$)
        for each $v_t$ tentative wrt $v$ do
            $(X^t_L,X^t_U)$ clocks modified in $LU$ wrt $v_t.LU$.
            $v_t.LU$ := $LU$; 
            propagate($v_t$,$X^t_L$,$X^t_U$)
        if ($v\not= v_{root}$) then
            ($X_L$,$X_U$) clocks modified in $LU$ wrt $v.LU$
            propagate($v$,$X_L$,$X_U$)

function newbounds($v$,$v'$,$X'_L$,$X'_U$)
 given a transition $v\to v'$, find new $LU$ bounds for $v$ knowing
 that $LU$ bounds for $v'$ have changed, and
 $X'_L$ are the clocks whose $L$ bound has changed,
 $X'_U$ are the clocks whose $U$ bound has changed 
\end{lstlisting}
}


Algorithm~\ref{fig:algorithm} presented below, computes a tree whose
nodes $v$ have four components: $v.q$ is a state of $\Aa$, $v.Z$ is a
zone, and $v.L$, $v.U$ are $LU$ bound functions.  Each node $v$ has a
successor $v_t$ for every transition $t$ of $\Aa$ from $(v.q,v.Z)$ resulting
in a non-empty zone. Some nodes will be marked tentative and not
explored further. After an exploration phase, tentative nodes will be
reexamined and some of them will be put on the stack for further
exploration. At every point the leaves of the tree constructed by the
algorithm will be of three kinds: tentative nodes, nodes on the stack,
nodes having no transition needed to be explored.

Our algorithm starts from the root node $v_{root}$ labeled with $q_0$
and $Z_0$: the initial state of $\Aa$, and the initial zone. We do not
set the $LU$ bounds for $v_{root}$ as this will be done by
$\verb|explore|$ procedure. The main loop repeatedly alternates an
\emph{exploration} and a \emph{resolution} phases until there are no
nodes to be explored. The exploration phase constructs a part of ASG
from a given node stopping at nodes that it considers tentative. 
During exploration $LU$ bounds of some nodes may be changed in order to 
preserve invariants I2 and I3.
The resolution phase examines tentative nodes and adds them to the
stack for exploration if condition G3 of the definition of ASG is no
longer satisfied.

At the call of the procedure $\verb|explore|(v)$, node $v$ is supposed
to have its state $v.q$ and zone $v.Z$ set but the value of $v.LU$ is
irrelevant. 
The zone  $v.Z$ is supposed to be not empty. We assume that the
constructed tree satisfies the invariants I1, I2, I3, but for the node
$v$ and the nodes on the stack. The goal of the $\verb|explore|$
procedure is to restore the invariant for $v$ and start exploration of
successors of $v$ if needed.

First, the procedure checks if $v.q$ is an accepting state. If so then
we know that this state is reachable since we assume that $v.Z$ is not
empty.  When $v.q$ is not accepting we consider two cases.  If there
exists a $non-tentative$ node $v''$ in the current tree such that
$v.q=v''.q$ and $v.Z\subseteq\ALU{v''.LU}(v''.Z)$) then $v$ is a
tentative node. The $LU$-bounds from $v''$ are copied to $v$, and
propagated so that invariant $I2$ is restored. This is the task of
$\verb|propagate|$ procedure that we describe below. If $v$ is not
covered then it should be explored. First, we compute its $LU$ bound
based on transitions that are disabled from $v$. The task of
function $\verb|disabled|$ is to calculate the $LU$ bounds so that the
invariant I1 holds. (The function is described in more detail in the next
section.) Then we propagate these bounds in order to restore
the invariant I2. Finally, we explore from every successor of $v$.

When $LU$ bounds in a node $v'$ are changed the invariant I2 should be
restored. For this the bounds are propagated by invoking
$\verb|propagate|$ procedure. For efficiency, the procedure is also
given the set of clocks $X'_L$ whose $L$ bound has changed, and the set
$X'_U$ of clocks whose $U$ bound has changed. The parent $v$ of $v'$
is taken and the transition from $v$ to $v'$ is examined. 
The function $\verb|newbounds|$ calculates new $LU$ bounds for a node
given the changes in its successor. This function is the core of our
algorithm and is the subject of the next section.
Here it is enough to assume that the new bounds are such that the
invariant I2 is satisfied. If the bounds of $v$ indeed change then
they should be copied to all nodes tentative with respect to $v$. This
is necessary to satisfy the invariant I3. Finally, the bounds are
propagated to the predecessor of $v$ to restore invariant I2.

The exploration phase terminates as in the \verb|explore| procedure
the bound functions in each node never decrease and are bounded.  They
are bounded because $\verb|newbounds|$ function never gives bounds
bigger than those obtained by static analysis (cf.\ Equation~\eqref{eq:static})

After exploration phase $LU$ bounds of tentative nodes may change. The
procedure \verb|resolve| is called to check for the consistency of
$tentative$ nodes. If $v$ is tentative w.r.t.\ $v'$ but $v.Z
\not\incl\ALU{v'.LU}(v'.Z)$ is not true anymore, $v$ needs to be
explored. Hence it is viewed as a new node, and put on the $stack$ for
further consideration in the function \verb|main|.

The algorithm terminates when either it finds and accepting state, or
there are no nodes to be explored and all tentative nodes remain
tentative. 
In the second case we can conclude that the constructed tree
represents an ASG, and hence no accepting state is reachable. 
Note that the overall algorithm should terminate as the bounds can
only increase and bounds in a node $(q,Z)$ are not bigger than the
bounds obtained for $q$ by static analysis (cf.\ Remark on
page~\pageref{rem:asg-finite}).

From the above discussion it follows that the algorithm returns
``\emph{empty}'' only when it constructs a complete ASG.  
The correctness of the algorithm then follows from
Theorem~\ref{thm:algo-correct}. 
\begin{proposition}\label{thm:alg correct}
  The algorithm always terminates. If for a given $\Aa$ the result is
  "not empty'' then $\Aa$ has an accepting run. Otherwise the
  algorithm returns empty after constructing ASG for $\Aa$ and not
  seeing an accepting state. 
\end{proposition}




\section{Controlling  $LU$-bounds}\label{sec:newbounds}
The notion of adaptive simulation graph (Definition~\ref{df:asg})
gives necessary conditions for the values of $LU$ bounds in every node. 
The invariant I1 tells that $LU$ bounds in a node should take into account the
the edges disabled from the node. 
The invariant I2 gives a lower bound on $LU$ with respect to the $LU$-bounds 
in successors of the node. 
Finally, I3 tells us that $LU$ bounds in a covered node should be 
not smaller than in the covering node. 
The algorithm from the last section implements a construction of ASG with 
updates of the bounds when the required by the invariant.

The three invariants sometimes allow for  much smaller $LU$-bounds
than that obtained by static analysis. 
A very simple example is when the algorithm does not 
encounter a node with a disabled edge. 
In this case all $LU$-bounds are simply $-\infty$, since no bound
is increased due to I1, and such bounds are not changed by propagation. 
When $LU$ bounds are $-\infty$, $\alu$ abstraction of a zone results
in the set of all valuations.
So in this case ASG can be just a subgraph of the automaton.
A more interesting examples of important gains are discussed  in the next
section.  

In this section we describe two central functions of the proposed
algorithm: $\verb|disabled|$ and $\verb|newbounds|$. 
The pseudo-code is presented in Algorithm~\ref{fig:bounds}.

The $\verb|disabled|$ function is quite simple. Its task is to 
restore the invariant I1. For this it chooses
from every disabled transition an atomic guard that makes it
disabled. Recall that we have assumed that every guard contains either
only lower bound constraints or only upper bound constraints. A
transition with only lower bound constraints cannot be
disabled. Hence a guard on a disabled transition must be a conjunction
of upper bound constraints. It can be shown that if such a
guard is not satisfied in a zone then there is one atomic constraint
that is not satisfied in a zone. Now it suffices to observe that if a
guard $x\leq d$ or $x< d$ is not satisfied in $Z$ then it is not
satisfied in $\alu(Z)$ when $U(x)=d$. This follows directly from the
definition of LU-simulation (Definition~\ref{def:lu-preorder}).

For the rest of this section we focus on the description of the function
$\verb|newbounds|(v,v',X'_L,X'_U)$. This function calculates
new $LU$-bounds for $v$, given that the bounds in $v'$ have changed. As
an additional information we use the sets of clocks $X'_L$ and
$X'_U$ that have changed their $L$-bound, and $U$-bound respectively,
in $v'$. This information makes the function $\verb|newbounds|$
more efficient since the new bounds depend only on the clocks in
$X'_L$ and $X'_U$.  The aim is to give bounds that are as small as
possible and at the same time satisfy invariant I2 from
Definition~\ref{df:asg}. 

Recall that we have assumed that every transition has either only
upper bound guards, or only lower bound guards and no resets
(cf.~page~\pageref{proviso}).
This assumption will simplify the $\verb|newbounds|$ function.
We will first consider the case of transitions with just an atomic guard or
with just a reset. Next we will put what we have learned together to
treat the general case.

\subsection{Reset}\label{subsec:reset}
Consider a transition $(q,Z)\tto_{R} (q',Z')$ for the set of clocks
$R$ being reset. So we have $Z'=\elapse{Z[R:=0]}$, i.e., we reset the
clocks in $R$ and let the time elapse. Suppose that we have updated
$L'U'$ and now we want our $\verb|newbounds|$ function to compute
$L_{new}U_{new}$.  We let $L_{new}U_{new}$ be the maximum of $LU$ and
$L'U'$ but for $\Lnew(x)=\Unew(x)=-\infty$ for $x\in R$. We want to
show that invariant I2 holds that is:
\begin{equation*}
  \aluNew(Z)[R:=0]\incl \aluP(Z[R:=0]).  
\end{equation*}
To prove this inclusion, take a valuation $v\in \aluNew(Z)$. By
definition there is a valuation $v'\in Z$ with $v\luNew v'$. We obtain
that $v[R:=0]\luP v'[R:=0]$ using directly
Definition~\ref{def:lu-preorder}. Indeed, for every clock in $R$, its values
in the two valuations are the same. For other clocks the required
implications hold since $v\luNew v'$ and moreover the bounds
$\Lnew\Unew$ and $L'U'$ are the same for these clocks.

\subsection{An abstract formula for atomic guard case} Consider a transition $(q, Z, LU)
\tto^{g} (q', Z', L'U')$. Suppose that we have updated $L'U'$ and now
we want our $\verb|newbounds|$ function to compute $L_{new}U_{new}$. In the
standard constant propagation algorithm, we would have set
$L_{new} U_{new}$ to be the maximum over $LU$, $L'U'$ and the constant
present in the guard. This is sufficient to maintain Invariant
2. However, it is not necessary to always take the guard $g$ into
consideration for the propagation.

Let $L_gU_g$ be the bound function induced by the guard $g$. In our
case where there is only one constraint, there is only one constant
associated to a single clock by $L_gU_g$.  Roughly, in order to
maintain Invariant 2, it suffices to take
\begin{align}\label{eq:pre}
  L_{new}U_{new} = 
  \begin{cases}
    \max(LU, L'U') & \text{if $LU \ge L_gU_g$ or } \\
    & \text{if $\sem{g} \incl \aluP(Z')$ or} \\
    & \text{if $Z \incl \aluP(Z')$} \\
    \max(LU,L'U',L_gU_g) & \text{otherwise}
  \end{cases}
\end{align}

To see why the above should maintain Invariant 2, look at the
transition with the new bounds:
\begin{align*}
  (q, Z, L_{new}U_{new}) \tto^{g} (q', Z', L'U')
\end{align*}
Clearly from the above definition, $L_{new}U_{new} \ge L'U'$.

Additionally, if $L_{new}U_{new} \ge L_gU_g$, that is, if the constant
in the guard is incorporated in $L_{new}U_{new}$, it is easy to show
using definition of simulation that $\post_g(\aluNew(Z)) \incl
\aluP(Z')$.

We now need to show the same for the cases when $L_{new} U_{new}$ does
not incorporate the constant in the guard. From the definition of the
$\pre$, this happens only if either $g \incl \aluP(Z')$ or if $Z \incl
\aluP(Z')$. Let us look closely at what $\post_g(\aluNew(Z))$ is.
\begin{align*}
  \post_g(\aluNew(Z)) = \elapse{\aluNew(Z) \cap \sem{g}}
\end{align*}
If $\sem{g} \incl \aluP(Z')$, then $\aluNew(Z) \cap \sem{g}$ would be
included in $\aluP(Z')$. As $\aluP$ is closed under time-elapse, we
will have $\post_g(\aluNew(Z)) \incl \aluP(Z')$. Similarly if $Z \incl
\aluP(Z')$, we will have $\aluP(Z) \incl \aluP(Z')$ and as
$L_{new}U_{new} \ge L'U'$, we have $\aluNew(Z) \incl \aluP(Z')$. It
follows that $\post_g(\aluNew(Z)) \incl \aluP(Z')$.

\subsection{A concrete algorithm for atomic guard case}
Since bound propagation is called very often in the main algorithm, we
need an efficient test for the inclusions in Formula~\eqref{eq:pre}.
The formula requires us to test inclusion w.r.t. $\alu$ between $Z$
and $Z'$ each time we want to do the $\pre$. Although this seems
complicated at the first glance, note that $Z'$ is a zone obtained by
a successor computation from $Z$. When we have only a guard in the
transition, we have $Z' = \elapse{Z \land g}$. This makes the
inclusion test lot more simpler. We will also see that it is not
necessary to consider the inclusion $\sem{g} \incl \aluP(Z')$.

Before proceeding, we need to look closer how zones are
represented. One standard way to represent zones is using
difference bound matrices (DBMs)~\cite{Dill:AVMFSS:1989}. We will
consider an equivalent representation in terms of distance graphs.

A \emph{distance graph}  has clocks as vertices, with an additional
special clock $x_0$ representing the constant $0$. 
For readability, we will often write $0$ instead of $x_0$. Between
every two vertices there is an edge with a weight of the form $(\lleq,
c)$ where $c\in \mathbb{Z}$ and $\lleq$ is either $\leq$ or $<$; or
$(\lleq, c)$ equals $(<, \infty)$. An edge $x\act{\lleq c} y$
represents a constraint $y-x\lleq c$: or in words, the distance from
$x$ to $y$ is bounded by $c$. An example of a distance graph is
depicted in Fig.~\ref{fig:dist-graph}.

\begin{figure}[!t]
  \centering
  \begin{tikzpicture}
  \begin{scope}[xshift=0cm]
    \node (0) at (0,0) {$0$};
    \node (x) at (2,0) {$x$};
    \node (y) at (4,0) {$y$};
    \begin{scope}[->,thick,>=stealth,bend angle=10]
      \tiny
      \draw (0) edge[bend left] node[above] {$(<,\infty)$} (x);
      \draw[rounded corners] (0.south) |- (2,-0.5) node[below]
      {$(<,2)$} -| (y.south);
      \draw (x) edge[bend left] node[below] {$(<,-4)$} (0);
      \draw (x) edge[bend left] node[above] {$(\le, -1)$} (y);
      \draw (y) edge[bend left] node[below] {$(<,\infty)$} (x);
      \draw[rounded corners] (y.north) |- (2,0.5) node[above]
      {$(<,\infty)$} -| (0.north);
    \end{scope}
  \end{scope}
\end{tikzpicture}

  \caption{Distance graph for the zone $(x-y \geq 1 \, \wedge \, y < 2
    \, \wedge \, x > 4)$.}
  \label{fig:dist-graph}
\end{figure}
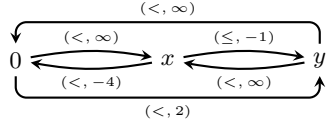

Let $\sem{G}$ be the set of valuations of clock variables satisfying
all the constraints given by the edges of $G$ with the restriction
that the value of $x_0$ is $0$. 

One can define an arithmetic and order over the weights $(\lleq, c)$ in an
expected manner~\cite{bengtsson2004timed}. We recall only the
definition of order that is most relevant for us here
\begin{description}
\item \emph{Order} $(\lleq_1,c_1) < (\lleq_2,c_2)$ if either $c_1 <
  c_2$ or ($c_1 = c_2$ and $\lleq_1 = < $ and $\lleq_2 = \le$).
\end{description}

A distance graph is in \emph{canonical form} if the weight of the edge
from $x$ to $y$ is the lower bound of the weights of paths from $x$ to
$y$. For instance, the distance graph shown in
Figure~\ref{fig:dist-graph} is not in canonical form as the weight of
the edge $x \xra{} y$ is $(\le, -1)$ whereas there is a path $x \xra{}
0 \xra{} y$ whose weight is $(<, - 2)$. To convert it to canonical
form, it is sufficient to change the weight of the edge $x \xra{} y$
to $(<, - 2)$. 

For two distance graphs $G_1$, $G_2$ which are not necessarily in
canonical form, we denote by $\min(G_1,G_2)$ the distance graph where
each edge has the weight equal to the minimum of the corresponding
weights in $G_1$ and $G_2$.  Even though this graph may be not in
canonical form, it should be clear that it represents intersection of
the two arguments, that is,
$\sem{\min(G_1,G_2)}=\sem{G_1}\cap\sem{G_2}$; in other words, the
valuations satisfying the constraints given by $\min(G_1,G_2)$ are
exactly those satisfying all the constraints from $G_1$ as well as
$G_2$.

A zone $Z$ can be identified with the distance graph in the canonical form
representing the constraints in $Z$. For two clocks $x$, $y$ we write
$Z_{xy}$ for the weight of the edge from $x$ to $y$ in this graph. A
special case is when $x$ or $y$ is $0$, so for example $Z_{0y}$
denotes the weight of the edge from $0$ to $y$.

We recall a theorem from~\cite{lics} that permits to handle $Z\incl
\aluP(Z')$ test efficiently.
\begin{theorem}
  Let $Z$, $Z'$ be two non-empty zones. Then $Z \not\incl \aluP(Z')$
  iff there exist two clocks $x,y$ such that:
\begin{equation}\label{eq:alu-inc}
  Z_{x0} \ge (\le, -U'_x) \text{ and $Z'_{xy} < Z_{xy}$ and $Z'_{xy} +
    (<,-L'_y) < Z_{x0}$}
\end{equation}
\end{theorem}
We are ready to proceed with our analysis. 
We distinguish two cases depending on whether the guard $g$ is of the
form $w \ggeq d$ or $w \lleq d$.

\paragraph{Lower bound guard:} When we have a lower bound guard, the
diagonals do not change during intersection and time-elapse.
Hence we have $Z'_{xy} = Z_{xy}$ when both $x$ and $y$ are non-zero
variables. This shows that (\ref{eq:alu-inc}) cannot be true when both
$x$ and $y$ are non-zero as the second condition is false. Yet again,
when $x$ is $0$, the second condition cannot be true as both $Z_{0y} =
Z'_{0y} = (<,\infty)$. It remains us to consider the single case when
$y$ is $0$. It boils down to checking if there exists a clock $x$ such
that:
\begin{align}\label{eq:incl-lower-bound-guard}
  Z_{x0} \ge (\le,-U'_x) \text{ and } Z'_{x0} < Z_{x0}
\end{align}
In words the above test asks if there exists a clock $x$ whose $x
\xra{\lleq -c} 0$ edge in $Z$ has reduced in $Z'$ and additionally the
edge weight $(\lleq, -c)$ in $Z$ satisfies either $c < U'_x$ or
$(\lleq,c) = (\le, U'_x)$. If such a clock exists, the definition of
$\pre$ in (\ref{eq:pre}) suggests that we need to check if $\sem{g}
\incl \aluP(Z')$.

Let us look at the distance graph of $\sem{g}$. It has an edge $w
\xra{ \lleq -d} 0$ and edges $x \xra{ \le 0} 0$ for all other clocks
$x$. All other edges are $\infty$. We now apply the inclusion test
(\ref{eq:alu-inc}) between this distance graph and $Z'$. Note that is
(\ref{eq:incl-lower-bound-guard}) is true, then there is a clock that
has $Z'_{x0} < Z_{x0}$. But as $Z_{x0} \le (\le,0)$, we will have
$Z'_{x0} < (\le, 0)$ which implies that $Z'_{x0} < \sem{g}_{x0}$. This
shows that if the inclusion between zones does not hold, then the
inclusion of the guard $g$ in $Z'$ also does not hold. Therefore
testing (\ref{eq:incl-lower-bound-guard}) is sufficient. This gives us
the following formula with the additional observation that
$Z'_{x0}$ can be only lesser than or equal to $Z_{x0}$.

\begin{align}\label{eq:pre-lower}
  L_{new}U_{new} =
  \begin{cases}
    \max(LU,L'U',L_gU_g) & \text{if $L(w)< d$ and } \\
    &  \exists x. ~\big(Z_{x0} \ge (\le,-U_x')\big) \land \big((Z'_{x0}< Z_{x0})\big)\\
    \max(LU, L'U') & \text{otherwise}
  \end{cases}
\end{align}

Also note that this can be easily extended to an incremental
procedure: whenever we add an extra clock to $U'$, then we need to
check only this clock. The above definition also suggests that
whenever only $L'$ is modified we don't have to check anything and
just propagate the new values of $L'$.

\paragraph{Upper bound guard:}
 When we have an upper bound guard, the
diagonals might change. However no edge $0 \xra{} x$ or $x \xra{} 0$
changes. Therefore we need to check (\ref{eq:alu-inc}) for two
non-zero variables $x$ and $y$.

In other words, among clocks $x$ that have a finite $U'$ constant and clocks
$y$ that have a finite $L'$ constant, we check if there is a diagonal $x \xra{} y$
that has strictly reduced in $Z'$ and additionally satisfies $Z'_{xy}
+ (<,L_y) < Z_{x0}$. Note that this also entails $\sem{g} \not \incl
\aluP(Z')$. This is because when $g$ is $w \lleq d$, we have
$\sem{g}_{xy} = (<, \infty)$ and $\sem{g}_{x0} = (\le,0)$ and hence
(\ref{eq:alu-inc}) becomes true when $Z$ is substituted with
$\sem{g}$. Therefore it is sufficient to check (\ref{eq:alu-inc}) for
non zero variables $x$ and $y$. This gives the following formula
function:
\begin{align}\label{eq:pre-upper-bound}
  L_{new}U_{new} =
  \begin{cases}
    \max(LU,L'U',L_gU_g) & \text{if $U(w)< d$ and }  \exists x,y.\ 
    \text{such that}\\
    &  Z_{x0}\ge (\leq,-U'_x) \text{ and }  \big( Z'_{xy} < Z_{xy} )  \text{ and }  \\
    & \quad  \big( Z'_{xy} + (<, -L_y') < Z_{x0} \big) \\
    \max(LU, L'U') & \text{otherwise}
  \end{cases}
\end{align}

This test can also be done incrementally. Each time we propagate, we
need to perform extra checks only when a new clock has got a finite
value for either $L'$ or $U'$. 

\paragraph{Upper bound and reset.}\label{up-reset}
Here we consider the case when we have guard and reset at the same
time. So we consider transition $Z\tto_{(w<d),R}Z'$. We will combine
the cases above since we will treat this transition as
\begin{equation*}
  Z\tto_{b:=0} Z^1\tto_{w<d} Z^2\tto_{R} Z^3\tto_{b<0} Z^4
\end{equation*}
Suppose we have $L'U'=L^4U^4$ that we want to propagate it back to
$Z$. Since $b$ is a clock introduced for technical reasons we can
assume that $L^4(b)=U^4(b)=-\infty$. We need to calculate the values
of changed edges in all the zones
\begin{itemize}
\item In $Z^1$ we get $Z^1_{b0}=0$, and $Z^1_{xb}=Z_{x0}$, and $Z^1_{bx}=\infty$.
\item In $Z^2$ we get $Z^2_{xy}=Z^1_{x0}+d+Z_{wy}$ (if this edge
  changes).
\item In $Z^3$ every edge stays the same but for the clocks that are
  reset. We have $Z^3_{v0}=0$, $Z^3_{xv}=Z_{x0}$,
  and $Z^3_{vx}=\infty$ for $v\in R$ and $x\not\in R$.

\item In $Z^4$ we get $Z^4_{xy}=Z_{x0}+Z^2_{by}$ if this edge
  changes. 
  \begin{itemize}
  \item Suppose $x\not\in R$. From the second item we know that
    $Z^2_{by}=(d+Z_{wy})$.  So
    $Z^4_{xy}=Z_{x0}+d+Z_{wy}=Z^2_{xy}$. This means that no edge
    changes from $Z^3$ to $Z^4$.
  \item Suppose $x\in R$ then
    $Z^4_{xy}=Z^3_{by}=Z^2_{by}=d+Z_{wy}$. Since $Z^3_{xy}=\infty$
    this edge necessarily changes.
  \end{itemize}
\end{itemize}
Because of the last item we see that we always take the guard $x<b$
into $U$. So $L^3U^3=L^4U^4[U_b=0]$. Now $L^2U^2=L^3U^3[R=-\infty]$. In order
to get $L^1U^1$ we apply the formula~\eqref{eq:pre-upper-bound} using
the knowledge what is the relation between $L'U'$ and $L^2U^2$:
\begin{align}\label{eq:pre-upper-reset-bound}
  L_{new}U_{new} =
  \begin{cases}
    \max(LU,L'U',L_gU_g) & \text{if $U(w)< d$ and }  \exists
    x,y\not\in R.\    \text{such that}\\
    &  Z^1_{x0}\ge (\leq,-U^3_x) \text{ and }  \big( Z^2_{xy} < Z^1_{xy} )
    \text{ and }  \\
    & \quad  \big( Z^2_{xy} + (<, -L'_y) < Z^1_{x0} \big) \\
   \max(LU,L'U',L_gU_g) &  \exists y\not\in R.\    \text{such that}\\
          & \quad  \big(d+Z_{wy} + (<, -L_y') < 0 \big) \\
    \max(LU, L'U') & \text{otherwise}
  \end{cases}
\end{align}
The second formula is the specialization of the first for the case of
$x=b$. So we see that we almost always take the $w<d$ guard.
Observe that the first condition implies the second since
$Z^2_{xy}=Z^1_{0x}+d+Z_{wy}$. So if $Z^2_{xy} + (<, -L'_y) < Z^1_{x0}$
then $Z^1_{0x}+d+Z_{wy}+(<-L'_y) < Z^1_{x0}$ which is equivalent to 
$d+Z_{wy}+(<-L'_y) < Z^1_{x0}-Z^1_{0x}$. But $Z^1_{x0}-Z^1_{0x}\leq 0$
since the zone is not empty.



\subsection{Implementation of the $newbounds$ function}

We consider a transition of the form
\begin{equation*}
  (q,Z,LU)\xra{g} (q',Z',L'U')
\end{equation*}
We suppose that $\verb|newbounds|$ function examines this
transition. The bounds $L'U'$ have been updated and now we 
determine how to update the bounds $LU$.
Let $X'_L$ be the set of clocks for which $L'$ bound has been
updated. Similarly $X'_U$ for $U'$ bounds.

We will define the new bounds for $(q,Z)$. So the node $(q,Z,LU)$ will
be changed to $(q,Z,\Lnew\Unew)$. Observe that the bounds can only
increase.

We have four cases depending on the type of the guard. The pseudocode
is presented in Algorithm~\ref{fig:bounds}

{
  \setlength{\columnsep}{35pt}
  \setlength{\columnseprule}{0.5pt}
\begin{lstlisting}[mathescape=true, float=tp, 
caption={\texttt{disabled} and \texttt{newbounds} functions},
label=fig:bounds]

function disabled($q$,$Z$)
  $L$:=$L_{-\infty}$; $U$:=$U_{-\infty}$;
  for every transition $t$ from $q$ disabled from $(q,Z)$ do
     choose an atomic guard $x\lleq d$ from the guard of $t$ 
     such that $Z\not\sat x\lleq d$ // guard of $t$ has only upper bound guards
     $U(x)$:=$\max(d,U(x))$
  return (L,U)



function newbounds($v$,$v$,$X'_L$,$X'_U$)
for every clock $x$ do
  if $x\in X'_L$ then $L_{new}(x)$:=$\max(L(x),L'(x))$ else $L_{new}(x)$:=$L(x)$;
  if $x\in X'_U$ then $U_{new}(x)$:=$\max(U(x),U'(x))$ else $U_{new}(x)$:=$U(x)$;

if transtion $v\to v'$ is a lower bound guard $\Land_{i=1\dots k} v_i\geq d_i$
  $E$:=$\set{(x,0) : x\in X'_U\ \text{ and }\ Z_{x0}\geq (\leq,-U'_x)\ \text{ and }\ Z'_{x0} < Z_{x0}}$
  while $E\not=\es$ do
    choose $d_i$ such that there is $(x,0)\in E$ with $-d_i+Z_{xv_i}=Z'_{x0}$;
    $L_{new}(v_i)$:=$\max(d_i,L_{new}(v_i));$
    $E$:=$E\setminus \set{(x,0) : d_i+Z_{xv_i}=Z'_{x0}}$

else if transtion $v\to v'$ is an upper bound guard $\Land_{i=1\dots k} w_i\leq e_i$
  $E$:=$\{(x,y) : x\in X'_U\ \text{ and }\ y\in X'_L,\ \text{ and }\  $
             $Z_{x0}\geq(\leq,-U'_x)\ \text{ and } Z'_{xy}< Z_{xy}\ \text{ and }\ Z'_{xy}+(<,-L_y)< Z_{x0}\}$;
  while $E\not =\es$ do
    choose $e_i$ such that there is $(x,y)\in E$ with  $e_i+Z_{w_iy}+Z_{x0}=Z'_{xy}$
    $U_{new}(w_i)$:=$\max(e_i,U_{new}(w_i))$;
    $E$:=$E\setminus\set{(x,y) : e_i+Z_{w_iy}+Z_{x0}=Z'_{xy}}$

else if transtion $v\to v'$ is a reset $R$
  for $x\in R$ do 
   $\Lnew(x)$=$L(x)$; $\Unew(x)$:=$U(x)$;

else if transition $v\to v'$ is an upper bound guard $\Land_{i=1\dots k} w_i\leq e_i$ and 
                            a reset $R$
  Fix some $r\in R$;
  $E$:=$\set{(r,y) : y\in X'_L\setminus R \text{ and } Z'_{ry}<(<,L'_y)}$;
  while $E\not=\es$ do
     choose $e_i$ such that there is $(r,y)\in E$ with $e_i+Z_{wy}=Z'_{ry}$
     $U_{new}(w_i)$:=$\max(e_i,U_{new}(w_i))$;
     $E$:=$E\setminus\set{(r,y): e_i+Z_{wy}=Z'_{ry}}$;

return($\Lnew$,$\Unew$);
\end{lstlisting}
}

\textbf{Lower bound guard} We consider a transition for the form
$(q,Z,LU)\xra{g_L} (q',Z',L'U')$ with $g_l\equiv \land_{i=1\dots k}
v_i\geq d_i$.
First, we set $L_{new}U_{new}$ to the maximum of $LU$ and $L'U'$;
notice that by the defintion of $X'_L$ and $X'_U$
we need to calculate maximum only for the clocks in these two sets. 
Then we establish the set of edges $E$ of the zone $Z'$ that have changed, and
that are relevant for the test~\eqref{eq:pre-lower}. 
The final loop decides which constraints should be taken to
increase $L$ bound. 
We take $d_i$ when it indeed determines some relevant edge from $E$. 
If we take $d_i$ then we update $L_{new}$, and remove from $E$ 
all edges that are set by $d_i$. 
This is because there may be another constraint that influences the same change
in $Z'$ and there is no point of taking it.

For the correctness proof let $g^1_L$ be the set of constraints that
have been taken and $g^2_L$ the constraints that have been omited. 
The transition $(q,Z,L_{new}U_{new})\tto_{g_L} (q',Z',L'U')$ can be
decomposed into $(q,Z,L_{new}U_{new})\tto_{g^1_L}
(q',Z^1,L'U')\tto_{g^2_L}(q',Z',L'U')$.
From the algorithm we know that all the edges from $E$ as in line $18$
are the same in $Z^1$ and $Z'$. 
Hence by formula~\eqref{eq:pre-lower} we get 
$\Post_{g^2_L}(\aluP(Z^1))\incl \aluP(Z')$. 
Since all the guards from $g^1_L$ are taken we get
$\Post_{g^1_L}(\aluNew(Z))\incl \aluP(Z^1)$.

\textbf{Upper bound guard} 
We consider a transition of the form 
$(q,Z,LU)\xra{g_U} (q',Z',L'U')$ with 
$g_l\equiv \land_{i=1\dots k} w_i\leq e_i$.
Let us explain Algorithm~\ref{fig:bounds} in this case.
As in the previous case we set $L_{new}U_{new}$ to the maximum of 
$LU$ and $L'U'$.
Next we calculate the set of edges $E$ that can influence taking a guard.
The final for loop considers a constraint one by one. 
When the constraint implies an edge in $E$ we take the constraing and
remove all the edges implied by it. 

The correctness proof is very similar to the previous case. Let
$g^1_U$ be the set of constarints that have been taken and $g^2_U$ the
constraints that have been omited. The transition
$(q,Z,L_{new}U_{new})\tto_{g_U} (q',Z',L'U')$ can be decomposed into
$(q,Z,L_{new}U_{new})\tto_{g^1_U}
(q',Z^1,L'U')\tto_{g^2_U}(q',Z',L'U')$.  From the algorithm we know
that all the edges from $E$ as in line $25$ are the same in $Z^1$ and
$Z'$.  Hence by formula~\eqref{eq:pre-upper-bound} we get
$\Post_{g^2_L}(\aluP(Z^1))\incl \aluP(Z')$.  Since all the guards from
$g^1_L$ are taken we get $\Post_{g^1_L}(\aluNew(Z))\incl \aluP(Z^1)$.

\textbf{Reset} The case of reset follows directly from the formula in
Section~\ref{subsec:reset}. 

\textbf{Upped bound and reset} This case follows directly from the
formula~\eqref{eq:pre-upper-reset-bound}.

\newcommand{\zonegraph}{zone graph \, } 

\section{Examples}
\label{sec:example}
In this section we will analyze behavior of our algorithm on some
examples in order to explain some of the sources of the gains reported
in the next section.


\subsection{All edges enabled}
\label{sec:all-edges-enabled}

Consider the automaton $\Aa_1$ shown in
Figure~\ref{fig:allenabled}. In the same figure, the
\zonegraph of $\Aa_1$ has been depicted. Note that the zone graph has
no edges disabled and hence is isomorphic to
the automaton. In such a case, observe that it is safe to abstract all
the zones by the true zone. The set of reachable states of the
automaton remain the same even after abstracting all zones to the true
zones. 

\begin{figure}
\centering
\begin{tikzpicture}[shorten >=1pt,node distance=2.5,on
      grid,auto]
     
      \begin{scope}
      \node[state,initial,initial text=] (q0) {$q_0$};
      \node[state] (q1) [right=of q0] {$q_1$};
      \node[state] (q2) [right=of q1] {$q_2$};
      \node[state] (q3) [right=of q2] {$q_3$};
      \end{scope}

      \begin{scope}[->]
        \draw (q0) edge node {\footnotesize $x \ge 5$} (q1);
        \draw (q1) edge node {\footnotesize $y \ge 5$} (q2);
        \draw (q2) edge node {\footnotesize $w \le 10$} (q3);
      \end{scope}

\end{tikzpicture}

\vspace{0.4cm}

\begin{tikzpicture}
\begin{scope}
\tikzstyle{every node}=[draw,fill=white,rounded corners=2mm]
\node (z0) at (-6,0) {{\scriptsize $q_0: (x = y = w \ge 0)$ }};
\node (z1) at (-2.5,0) {{\scriptsize $q_1: (x = y = w \ge 5)$ }};
\node (z2) at (1,0) {{\scriptsize $q_2: (x = y = w \ge 5)$ }};
\node (z3) at (4.5,0) {{\scriptsize $q_3: (x = y = w \ge 5)$ }};
\end{scope}

\begin{scope}[double distance = 1.5pt, very thin]
\draw[->] (-7.7,0) -- (z0);
\draw[->, double] (z0) --  node[above] {\tiny $x \ge 5$} (z1);
 \draw[->, double] (z1) --  node[above] {\tiny $y \ge 5$}
 (z2);
 \draw[->, double] (z2) --  node[above] {\tiny $w \le 10$}
 (z3);
\end{scope}
\end{tikzpicture}
\caption{$\Aa_1$: all edges enabled in the zone graph}
\label{fig:allenabled}

\end{figure}
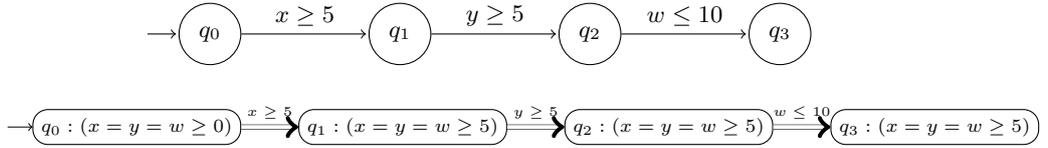


Algorithm~\ref{fig:algorithm} is able to incorporate this
phenomenon. Initially all the constants are $-\infty$ and hence the
$\alu$ abstraction of each zone would give the true zone. The
algorithm starts propagating finite $LU$-constants only when it
encounters a disabled edge during exploration. In particular, if there
are no edges disabled, all the constants are kept $-\infty$. We will
now see an example where this property of the propagation yields
exponential gain over the static analysis method and the on-the-fly
constant propagation procedure.


\begin{figure}[tbh]
  \centering
  \includegraphics[scale=.7]{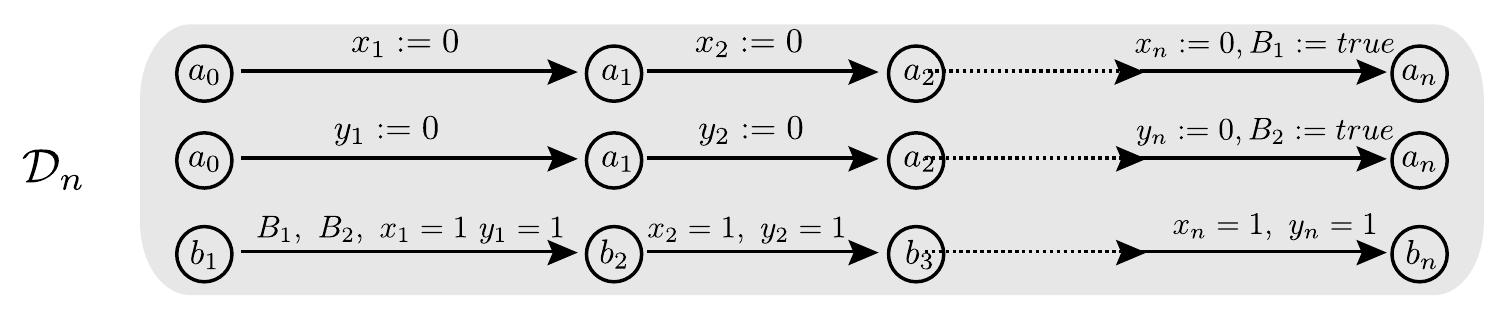}
  \caption{Automaton $\Dd_n$}
  \label{fig:niebert-enabled}
\end{figure}
Consider the automaton $\Dd_n$ shown in
Figure~\ref{fig:niebert-enabled}. This is slightly modified from the
example given in~\cite{LugiezNiebertZ05}. We have changed all guards to check
for an equality. It is a parallel composition of
three components. Automaton $\Dd_n$ has $2n$ clocks: $x_1, \dots, x_n$
and $y_1, \dots, y_n$. The first two components respectively reset the $x$-clocks
and $y$-clocks. The third component can be fired only after the first
two have reached their $a_n$ states. The states of the product
automaton $\Dd_n$ are of the form $(a_i, a_j, b_0)$ and $(a_n, a_n,
b_k)$ where $i, j , k \in \{ 0, \dots, n\}$. In all, there are $(n+1)^2
+ n$ states in the product automaton. Let us assume that no state is
accepting so that any algorithm that explores this automaton should
explore the entire zone graph.

Clearly, all the transitions can be fired if no time elapses in the
states $(a_i, a_j, b_0)$  for $i, j \in {1, \dots, n-1}$ and exactly
one time unit elapses in $(a_n, a_n, b_0)$. Therefore, the zone graph
of $\Dd_n$ should have no edges disabled which implies that
the $LU$-constants given by Algorithm~\ref{fig:algorithm} in each node
are $-\infty$. The number of uncovered nodes in the ASG obtained would
be the same as the number of states. 

\paragraph{Static analysis:}
However, the static analysis procedure would give $L = U = 1$ for
every clock. We will now see that this would yield a zone graph with
at least $2^n$ nodes. 

\begin{figure}
\centering

\begin{tikzpicture}

\begin{scope}[scale=0.8]
\draw[->,thick, >=stealth] (0,0) -- (4,0);
\draw[->,thick, >=stealth] (0,0) -- (0,4);

\node at (-0.3, -0.3) {\footnotesize $0$};

\node at (4.1, -0.3) {\footnotesize $x$};
\node at (-0.3, 4.1) {\footnotesize $y$};

\draw[very thick] (2,0) -- (2,4);
\draw[very thick] (0,2) -- (4,2);

\node at (2, -0.3) {\scriptsize $L_x = U_x = 1$};
\node [left] at (0, 2) {\scriptsize $L_y = U_y = 1$};

\fill[gray, nearly transparent] (0,0) -- (4,4) -- (0,4) -- cycle;
\draw[thin] (0,0) -- (4,4);

\fill[pattern=dots] (2,2) -- (4,2) -- (4,4) -- cycle;

\begin{scope}[xshift=1cm, yshift=-1.2cm]
\draw[fill=gray,nearly transparent] (0,0) rectangle (0.3,0.3);
\node [right] at (0.3, 0.15) {\footnotesize $Z_1$};

\draw[fill=gray, nearly transparent] (0,-0.5) rectangle (0.3,-0.2);
\draw[fill, pattern=dots] (-0.7, -0.5) rectangle (-0.4, -0.2);
\node at (-0.2, -0.35) {\footnotesize $\cup$};
\node [right] at (0.3, -0.35) {\footnotesize $\alu(Z_1)$};
\end{scope}

\end{scope}

\begin{scope}[xshift = 6cm,scale=0.8]
\draw[->,thick, >=stealth] (0,0) -- (4,0);
\draw[->,thick, >=stealth] (0,0) -- (0,4);

\node at (-0.3, -0.3) {\footnotesize $0$};

\node at (4.1, -0.3) {\footnotesize $x$};
\node at (-0.3, 4.1) {\footnotesize $y$};

\draw[very thick] (2,0) -- (2,4);
\draw[very thick] (0,2) -- (4,2);

\node at (2, -0.3) {\scriptsize $L_x = U_x = 1$};
\node [left] at (0, 2) {\scriptsize $L_y = U_y = 1$};

\fill[gray, nearly transparent] (0,0) -- (4,4) -- (4,0) -- cycle;
\draw[thin] (0,0) -- (4,4);

\fill[pattern=dots] (2,2) -- (2,4) -- (4,4) -- cycle;

\begin{scope}[xshift=1cm, yshift=-1.2cm]
\draw[fill=gray,nearly transparent] (0,0) rectangle (0.3,0.3);
\node [right] at (0.3, 0.15) {\footnotesize $Z_2$};

\draw[fill=gray, nearly transparent] (0,-0.5) rectangle (0.3,-0.2);
\draw[fill, pattern=dots] (-0.7, -0.5) rectangle (-0.4, -0.2);
\node at (-0.2, -0.35) {\footnotesize $\cup$};
\node [right] at (0.3, -0.35) {\footnotesize $\alu(Z_2)$};
\end{scope}

\end{scope}

\end{tikzpicture}
\caption{Zones indistinguishable by $\alu$}
\label{fig:alu-example-1}
\end{figure}

Consider Figure~\ref{fig:alu-example-1} that shows two zones $Z_1$ and
$Z_2$ and their $\alu$ abstractions when $L = U = 1$ for both the
clocks $x$ and $y$. Zone $Z_1$ is given by all valuations that satisfy
$x \le y$. Similarly zone $Z_2$ is given by all valuations that
satisfy $x \ge y$. Observe that $Z_1$ and $Z_2$ are incomparable with
respect to $\alu$, that is, $Z_1 \not \incl \alu(Z_2)$ and $Z_2 \not
\incl \alu(Z_1)$. 

In our example of the automaton $\Dd_n$, if in a path, $x_1$ is reset before
$y_1$ then in the state $(a_n, a_n, b_0)$ we would have a zone that
entails $y_1 \le x_1$. Similarly if $y_1$ is reset before $x_1$, then
the zone would entail $x_1 \le y_1$. In each of these paths to $(a_n,
a_n, b_0)$ clock $x_2$ could be reset either before or after $y_2$ and
so on for each $x_i$. There are at least $2^n$ paths leading to $(a_n,
a_n, b_0)$ each of them giving a different zone depending on the order
of resets. Note that two zones are incomparable if a projection onto 2
clocks are incomparable. By the argument in the previous paragraph,
each of the mentioned zones would be incomparable with respect to the
other. Therefore there are at least $2^n$ uncovered nodes with state
$(a_n, a_n, b_0)$. 

\paragraph{$\alu$,otf:}

As all the edges are enabled, the constant propagation algorithm would
explore a path up to $(a_n, a_n, b_n)$. This would therefore give $L =
U = 1$ for each clock, similar to static analysis. So in this case too
there would be at least $2^n$ uncovered nodes in the reachability tree
obtained.

\subsection{Presence of disabled edges}

Consider the automaton $\Aa_2$ in Figure~\ref{fig:onedisabled}. One
can see that the last transition with the upper bound is not
fireable. The cause of the edge being disabled is because the value of
$w$ in all the valuations of $Z_3$ is bigger than $1$. The cause of
this increase is the first lower bound guard $x \ge 5$. At $q_1$
itself, all the valuations have $w \ge 1$. As $w$ is never reset in
the automaton, there is no way $w$ can get lesser than $1$ after
passing this guard. Note that the guards $y \ge 5$ and $z \ge 100$ do
not play a role at all in the edge being disabled. Even if they had
not been there, the edge would be disabled.

\begin{figure}
\centering
\begin{tikzpicture}[shorten >=1pt,node distance=2.5,on
      grid,auto]
     
      \begin{scope}
      \node[state,initial,initial text=] (q0) {$q_0$};
      \node[state] (q1) [right=of q0] {$q_1$};
      \node[state] (q2) [right=of q1] {$q_2$};
      \node[state] (q3) [right=of q2] {$q_3$};
      \node[state] (q4) [right=of q3] {$q_4$};
      \end{scope}

      \begin{scope}[->]
        \draw (q0) edge node {\footnotesize $x \ge 5$} (q1);
        \draw (q1) edge node {\footnotesize $y \ge 5$} (q2);
        \draw (q2) edge node {\footnotesize $z \ge 100$} (q3);
        \draw (q3) edge node {\footnotesize $w  \le 2$} (q4);
      \end{scope}

\end{tikzpicture}

\vspace{0.4cm}

\begin{tikzpicture}
\begin{scope}
\tikzstyle{every node}=[draw,fill=white,rounded corners=2mm]
\node (z0) at (-6,0) {{\scriptsize $q_0: (x = y = w \ge 0)$ }};
\node (z1) at (-2.5,0) {{\scriptsize $q_1: (x = y = w \ge 5)$ }};
\node (z2) at (1,0) {{\scriptsize $q_2: (x = y = w \ge 5)$ }};
\node (z3) at (4.5,0) {{\scriptsize $q_3: (x = y = w \ge 100)$ }};
\end{scope}

\begin{scope}[double distance = 1.5pt, very thin]
\draw[->] (-7.7,0) -- (z0);
\draw[->, double] (z0) --  node[above] {\tiny $x \ge 5$} (z1);
 \draw[->, double] (z1) --  node[above] {\tiny $y \ge 5$}
 (z2);
 \draw[->, double] (z2) --  node[above] {\tiny $y \le 100$}
 (z3);
\draw[->, double] (z3) -- (6.5,0);
\end{scope}
\node at (6.2, 0) {\large $\times$};
\end{tikzpicture}
\caption{$\Aa_2$: One edge disabled}
\label{fig:onedisabled}

\end{figure}


 We want to capture this
scenario by saying that at $q_0$ the relevant constants are: $L_0(x) =
5$ and $U_0(x) = 1$ and the rest are $-\infty$. One can verify that Algorithm
~\ref{fig:algorithm} would give exactly these constants. The static analysis
algorithm or the constant propagation would give additionally $L(y) =
5 $ and $L(z) = 100$, which we have seen are unnecessary. This way, we
get smaller constants and hence bigger abstract zones. 

We will now see that this pruning can sometimes lead to an exponential
gain. We will modify the example $\Dd_n$ of
Section~\ref{sec:all-edges-enabled}. 

\begin{figure}[tbh]
  \centering
  \includegraphics[scale=.7]{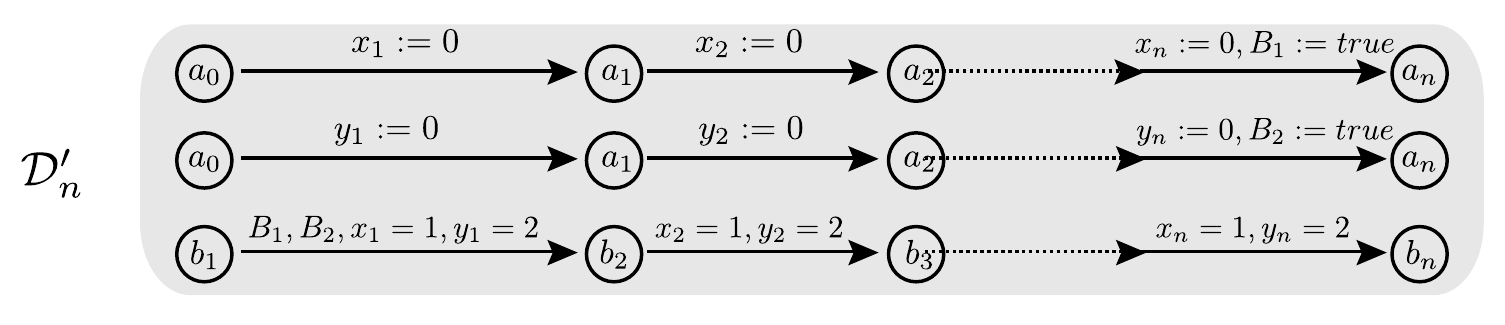}
  \caption{Automaton $\Dd'_n$}
  \label{fig:niebert-disabled}
\end{figure}

Let $\Dd'_n$ be the automaton shown in
Figure~\ref{fig:niebert-disabled}. It is the same as $\Dd_n$ except
that now every guard involving $y$-clock is $y == 2$. Starting
from a node $((a_n,a_n, b_0), Z, LU)$, it is possible to reach a node
with state $(a_n, a_n, b_n)$ only if $Z$ entails $x_i \le y_i$ for all
$i$. If ``fortunately'', the order of exploration of the resets leads us to such a
zone $Z$, then this path would yield no constants and hence the
abstraction would give the true zone. Due to this there would not be
any more exploration from $(a_n, a_n, b_0)$ and we would have the
number of uncovered nodes equal to number of states of automaton.

If it is not the case, then there is an $i$ such that $y_i \le x_i$
and for all $j < i$, $x_i \le y_i$. Therefore, the path can be
taken till $b_{i-1}$ after which the transition gets disabled because
we check for $y_i \ge 2$ and $x_i \le 1$. The disabled edge gives the
constant $U(x_i) = 1$ and the propagation algorithm additionally
generates $L(y_i) = 2$ and propagates these two backwards. These are
the relevant guards that cause the disabled edge. Since these are the
only constants, in the future, exploration will not occur from a node
$((a_n, a_n, b_0), Z', L'U')$ if $Z'$ satisfies $x_i \le y_i$ as they
will be covered. There will be at most $n$ uncovered nodes with the
state $(a_n, a_n, b_0)$ and hence the total number of uncovered nodes
will be in size quadratic in $n$.

\paragraph{Static analysis:}
The static analysis procedure would give $L = U = 2$ for all
$y$-clocks and $L = U = 1$ for all $x$-clocks. A similar argument as
in Section~\ref{sec:all-edges-enabled} would show at least $2^n$
uncovered nodes with state $(a_n, a_n, b_0)$. 

\paragraph{$\alu$,otf:}
The otf bounds algorithm could work slightly different from the
previous case. The constants generated depend on the first path. If
the first path leads up to $(a_n, a_n, b_n)$ then there are constants
generated for all clocks. Then, the zone cannot cover any of the
future zones that appear at $(a_n, a_n,b_n)$. A depth-first search
algorithm would clearly then be exponential.  Otherwise, if the path gets cut at $b_{k-1}$ 
constants are generated for all clocks $x_1, y_1, \dots, x_k, y_k$. In
this case, at least $2^k$ nodes at $(a_n,a_n,b_0)$ need to be
distinguished. 



\section{Experiments}
We report experiments in Table~\ref{table:experiments} for classical
benchmarks from the literature. The first two columns compare
UPPAAL~4.1.13 with our own implementation of UPPAAL's algorithm
(\textsf{$\ExtraLUp$,sa}). We have taken particular care to ensure
that the two implementations deal with the same model and explore it
in the same way. However, on the last example (Stari), we did not
manage to force the same search order in the two tools. 

The last two algorithms are using bounds propagation. In the
third column (\textsf{$\alu$,otf}), we report the results for the
algorithm in~\cite{Herbreteau:FSTTCS:2011} that propagates the bounds
from every transition (enabled or disabled) that is encountered during
the exploration of the zone graph. Since this algorithm only considers
the bounds that are reachable in the zone graph, it generally visits less nodes 
than UPPAAL's algorithm. The last column
(\textsf{$\alu$,disabled}) corresponds to the algorithm introduced in
this paper. It propagates the bounds that come from the disabled
transitions only. As a result it generally outperforms the other
algorithms. The actual implementation of our algorithm is slightly more
sophisticated than presented in Algorithm~\ref{fig:algorithm}. Similarly to 
UPPAAL, it uses a Passed/Waiting list instead of a stack. 
The implemented algorithm is presented in
Appendix~\ref{sec:new-algorithm-detailed}.

\begin{table}[t]
  \scriptsize
  \centering
  \begin{tabular}{|c|r||r|r||r|r||r|r||r|r|}
    \hline
    Model
    & nb. of
    & \multicolumn{2}{|c||}{UPPAAL (\texttt{-C})}
    & \multicolumn{2}{|c||}{\textsf{$\ExtraLUp$,sa}}
    & \multicolumn{2}{|c||}{\textsf{$\alu$,otf}}
    & \multicolumn{2}{|c|}{\textsf{$\alu$,disabled}}
    \\
    & clocks
    & nodes & sec.
    & nodes & sec.
    & nodes & sec.
    & nodes & sec.
    \\
    \hline
    $\mathcal{D}_7''$
    & 14
    & 18654 & 11.6
    & 18654 & 8.1
    & 213 & 0.0
    & 72 & 0.0
    \\
    $\mathcal{D}_8''$
    & 16
    & & 
    & &
    & 274 & 0.0
    & 90 & 0.0
    \\
    $\mathcal{D}_{70}''$
    & 140
    & &
    & &
    & &
    & 5112 & 1.9
    \\
    \hline
    CSMA/CD 10
    & 11
    & 120845 & 1.9
    & 120844 & 6.3
    & 78604 & 6.1
    & 74324 & 6.1
    \\
    CSMA/CD 11
    & 12
    & 311310 & 5.4
    & 311309 & 16.8
    & 198669 & 16.1
    & 188315 & 15.9
    \\
    CSMA/CD 12
    & 13
    & 786447 & 14.8
    & 786446 & 44.0
    & 493582 & 41.8
    & 469027 & 40.9
    \\
    \hline
    FDDI 50
    & 151
    & 12605 & 52.9
    & 12606 & 29.4
    & 5448 & 14.7
    & 401 & 0.8
    \\
    FDDI 70
    & 211
    & &
    & &
    & &
    & 561 & 2.7
    \\
    FDDI 140
    & 421
    & &
    & &
    & &
    & 1121 & 37.6
    \\
    \hline
    Fischer 9
    & 9
    & 135485 & 2.4
    & 135485 & 8.9
    & 135485 & 11.4
    & 135485 & 24.7
    \\
    Fischer 10
    & 10
    & 447598 & 10.1
    & 447598 & 34.0
    & 447598 & 42.8
    & 447598 & 98.1
    \\
    Fischer 11
    & 11
    & 1464971 & 40.4
    & 1464971 & 126.8
    & &
    & &
    \\
    \hline
    Stari 2
    & 7
    & 7870 & 0.1
    & 6993 & 0.4
    & 5779 & 0.4
    & 5113 & 0.5
    \\
    Stari 3
    & 10
    & 136632 & 1.7
    & 113958 & 9.4
    & 82182 & 8.2
    & 53178 & 7.8
    \\
    Stari 4
    & 13
    & 1323193 & 26.2
    & 983593 & 109.0
    & 602762 & 84.9
    & 342801 & 65.7
    \\
    \hline
  \end{tabular}

  \caption{Comparison of reachability algorithms: number of visited
    nodes and running time. For each model and each algorithm, we kept
    the best of depth-first search and breadth-first
    search. Experiments done on a MacBook with 2.4GHz Intel Core Duo
    processor and 2GB of memory running MacOS X 10.6.8. Missing
    numbers are due to time out (150s) or memory out (1Gb).}
  \label{table:experiments}
\end{table}

The results show a huge gain on two examples: $\mathcal{D}''$ and
FDDI. $D''_n$ corresponds to the automaton $\mathcal{D}_n$ in
Fig.~\ref{fig:onedisabled} where the tests $x_k=1, y_k=1$ have been
replaced by $(0 < x_k \le 1), (1 < y_k \le 2)$. While it was easier in
Section~\ref{sec:example} to analyze the example with equality tests,
we wanted here to show that the same performance gain occurs also when
static $L$ bounds are different from static $U$ bounds.
The number of nodes visited by algorithm \textsf{$\alu$,disabled}
exactly corresponds to the number of states in the timed
automaton. The situation with the FDDI example is similar: it has only
one disabled transition. The other three algorithms take useless clock
bounds into account. As a result they quickly face a combinatorial
explosion in the number of visited nodes. We managed to analyze
$\mathcal{D}''_n$ up to $n=70$ and FDDI up to size 140 despite the
huge number of clocks

Fischer example represents the worst case scenario for our
algorithm. Dynamic bounds calculated by algorithms
\textsf{$\alu$,otf} and \textsf{$\alu$,disabled} turn out to be the same
$LU$-bounds given by static analysis.

The last two models, CSMA/CD and Stari~\cite{Bozga:CHARME:1999} show
the average situation. 
The interest of Stari is that it is a very complex example with both a
big discrete part and big continuous part. 
The model is exactly the one presented in op. cit. but for a fixed
initial state. 
Algorithm \textsf{$\alu$,disabled} discards
many clock bounds by considering disabled transitions only. This leads
to a significant gain in the number of visited nodes at a reasonable
cost.


\section{Conclusions}
We have pursued an idea of adapting abstractions while searching
through the reachability space of a timed automaton. 
Our objective has been to obtain as low $LU$-bounds as possible
without sacrificing practicability of the approach.
In the end, the experimental results show that algorithm
\textsf{$\alu$,disabled} improves substantially the state-of-the art
algorithms for the reachability problem in timed automata.

At first sight, a more refined approach would be to work with
constraints themselves instead of $LU$-abstractions. Following the
pattern presented here, when encountering a disabled transition, one
could take a constraint that makes it disabled, and then propagate
this constraint backwards using, say, weakest precondition
operation. A major obstacle in implementing this approach is the
covering condition, like G3 in our case. When a node is covered, a
loop is formed in the abstract system. To ensure soundness, the
abstraction in a covered node should be an invariant of this loop. A
way out of this problem can be to consider a different covering
condition as proposed by McMillan~\cite{McMillan:CAV:2006}, but then this condition
requires to develop the abstract model much more than we do.  So from
this perspective we can see that $LU$-bounds are a very interesting
tool to get a loop invariant cheaply, and offer a good balance between
expressivity and algorithmic effectiveness.

We do not make any claim about optimality of our backward propagation
algorithm. For example,  one can see that it gives different results
depending on the order of treating the constraints. Even for a single
constraint, our algorithm is not optimal in a sense that there are examples
when we could obtain smaller $LU$-bounds. At present we do not know
if it is possible to compute optimal $LU$-bounds efficiently.
In our opinion though, it will be even more interesting to look at ways
of cleverly rearranging transitions of an automaton to limit bounds propagation
even further. 
Another promising improvement is to introduce some partial order
techniques, like parallelized interleaving from~\cite{Morbe:CAV:2011}.
We think that the propagation mechanisms presented here are well
adapted to such methods.



\bibliographystyle{alpha}
\bibliography{m}

\appendix

\newcommand{\coveredby}[0]{\ensuremath{\lhd}}
\newcommand{\propagatesto}[0]{\ensuremath{\leadsto}}

\section{Implementation of Algorithm \textsf{$\alu$,disabled}}
\label{sec:new-algorithm-detailed}

Algorithm~\ref{alg:uppaal} gives an overview of UPPAAL's algorithm. It
takes as input\footnote{The implementation builds the zone graph
  on-the-fly from a timed automaton taken as input.} a zone graph and
searches for a reachable accepting state. When a new node is expanded
(l.~11), it is first checked if it is covered by a visited node
(l.~15). If so, then it does not need to be explored. If not, all the nodes
that are covered by the new node are removed (l.~20-21) before the new
node is inserted to save memory and time.

In order to ensure the termination of the algorithm, the zones are
abstracted with an extrapolation operator
(e.g. $\ExtraLUp$\cite{Behrmann:STTT:2006}) that
guarantees a finite number of abstracted zones. The abstraction
parameters are clock bounds $LU$. They are obtained by a static analysis of
the timed automaton\cite{Behrmann:TACAS:2003}.

\begin{lstlisting}[mathescape=true,caption={UPPAAL's algorithm.},
label=alg:uppaal]
$P$ := $\emptyset$  // Passed list (visited nodes)
$W$ := $\emptyset$  // Waiting list  (W is included in P)

function main(): // input: zone graph $ZG$=($v_0$,$V$,$\xra{}$)
  insertPW($v_0$)
  while ($W$ is not empty) do
    pick a node $v$ from $W$
    if ($v.q$ is accepting)
      return ``not empty''
    for each transition $v \xra{} v'$ in $ZG$ do
      insertPW($v'$)
  return ``empty''

function insertPW($v$):
  if ($\exists v' \in P$ s.t. $v.q=v'.q$ and $v.Z \subseteq v'.Z$)
    // don't add v as it is covered by v'
    return
  else
    // remove all nodes v' covered by v
    for each $v' \in P$ s.t. $v'.q=v.q$ and $v'.Z \subseteq v.Z$ do
      remove $v'$ from $P$ and from $W$
    // insert v
    insert $v$ in $P$ and in $W$
\end{lstlisting}

Our algorithm \textsf{$\alu$,disabled} is built on top of UPPAAL's
algorithm. It is depicted in Algorithm~\ref{alg:alu-disabled}. The main
difference is that it computes dynamic $LU$-bounds that are used to
stop the exploration earlier. The dynamic bounds are used in l.~15. We
avoid exploring a node if it is covered by a visited node
w.r.t. dynamic bounds and abstraction $\alu$. If the node is not
covered, then its bounds are updated w.r.t. the transitions that are
disabled from that node (l.~21) and the node is explored (l.~24).

The algorithm computes an adaptive simulation graph $\propagatesto$
(see Definition~\ref{df:asg}) and a covering relation
$\coveredby$. The tentative nodes in Definition~\ref{df:asg} are the
nodes $v$ that are covered by some node $v'$, that is: $v \coveredby
v'$. The algorithm propagates the bounds and it updates
$\propagatesto$ and $\coveredby$ in order to maintain the invariants
in Definition~\ref{df:asg}. 

As the bounds are propagated over the graph $\propagatesto$, some
covering edge $v' \coveredby v$ may become invalid. This is checked in
line~50. When the bounds in $v'$ have to be updated from the bounds in
the covering node $v$, it is first checked if $v'$ is still
covered by $v$. If it is not the case, $v'$ is put in the list of
waiting nodes and it will be considered again later.

The propagation of clock bounds relies on function \texttt{newbounds}
given in Algorithm~\ref{fig:bounds}.

\begin{lstlisting}[mathescape=true,escapechar=£,caption={Algorithm
\textsf{$\alu$,disabled}.}, label=alg:alu-disabled]
// Assumptions: no lower bound atomic guards $d \lleq x$ in invariants
//              no atomic guard $x < 0$

$P$ := $\emptyset$  // Passed list (visited nodes)
$W$ := $\emptyset$  // Waiting list  (W is included in P)
$\coveredby$ := $\emptyset$ // Covering relation wrt dynamic bounds
$\propagatesto$ := $\emptyset$ // Propagation relation

function main(): // input: zone graph $ZG$=($v_0$,$V$,$\xra{}$)  
  insertPW($v_0$)
  while ($W$ is not empty) do
    pick a node $v$ from $W$
    if ($v.q$ is accepting)
      return ``not empty''
    if ($\exists v' \in (P \setminus W)$ uncovered st $v.q=v'.q$ and $v.Z \subseteq \ALU{v'.LU}(v'.Z)$)
      add $v \coveredby v'$ and $v' \propagatesto v$
      $v.LU$ := $v'.LU$
      $(X_L,X_U)$ := bounds modified during the copy
      propagate($v$, $X_L$, $X_U$)
    else
      $v.LU$ := disabled($v$)
      $(X_L,X_U)$ := active clocks in $v.LU$
      propagate($v$, $X_L$, $X_U$)
      for each transition $v \xra{} v'$ in $ZG$ do
        add $v' \propagatesto v$
        insertPW($v'$)
  return ``empty''

function insertPW($v$):
  if ($\exists v' \in P$ s.t. $v.q=v'.q$ and $v.Z \subseteq v'.Z$)
    // v is covered by v' wrt static bounds
    replace all $v \propagatesto v''$ by $v' \propagatesto v''$
  else
    // remove all nodes v' covered by v wrt static bounds
    for each $v' \in P$ s.t. $v'.q=v.q$ and $v'.Z \subseteq v.Z$ do
      remove $v'$ from $P$ and from $W$
      replace all $v' \propagatesto v''$ by $v \propagatesto v''$
      remove all $v'' \propagatesto v'$
      if ($\exists v'' \in P$ st $v' \coveredby v''$)
        remove $v' \coveredby v''$
      else
        for each $v'' \in P$ st $v'' \coveredby v'$ do
          remove $v'' \coveredby v'$
          insert $v''$ in $W$
    insert $v$ in $P$ and in $W$

function propagate($v$, $X_L$, $X_U$):
  for each $v'$ st $v \propagatesto v'$ do
    if ($v' \coveredby v$) // propagation due to a covering edge
      if ($v'.Z \subseteq \ALU{v.LU}(v.Z)$) // v' still covered by v
        $v'.LU$ := $v.LU$
        $(X'_L,X'_U)$ := bounds modified during the copy
      else // v' is not covered by v anymore
        $v'.LU$ := $x \mapsto -\infty$; $(X'_L,X'_U)$ := $(\emptyset, \emptyset)$
        insert $v'$ in $W$
    else // propagation due to a transition in ZG
      let $t$ be the transition $q' \xra{t} q$ that corresponds to $v \propagatesto v'$
      $(g_l,g_u,R)$ := decompose($t$)
      $(L_tU_t,X_L^t,X_U^t)$ := backwardLU($Z'$, $g_l$, $g_u$, $R$, $v.LU$, $X_L$, $X_u$)
      $v'.LU$ := max($v'.LU$, $L_tU_t$)
      $(X'_L,X'_U)$ := bounds modified by maximization
    if ($X'_L \neq \emptyset$ or $X'_U \neq \emptyset$)
      propagate($v'$, $X'_L$, $X'_U$)

function disabled($v$):
  $L$ := $x \mapsto -\infty$; $U$ := $x \mapsto -\infty$
  for each transition $t$ from $v.q$ that is disabled from $v.Z$ do
    $(g_l,g_u,R)$ := decompose($t$) // lower bounds,upper bounds,reset
    choose an atomic guard $w \lleq d$ in $g_u$ disabled from $v.Z \land g_l$
    $L_d$ := $x \mapsto -\infty$; $U_d$ := $w \mapsto d, \ x \mapsto
    -\infty \ (x \neq w)$
    $(L_tU_t, X_L, X_u)$ := backwardLU($v.Z$, $g_l$, $true$, $\emptyset$, $L_dU_d$, $\emptyset$, $\{x\}$)
    $LU$ := max($LU$, $L_tU_t$)
  return $LU$

function decompose($t$):
  let $t \, = \, (I,g,R,I')$ // src inv, guard, reset, tgt inv
  $g'$ := $g \land I$
  add to $g'$ all the atomic guard $x \lleq d$ from $I'$ st $x \not\in R$
  let $g'_l$ be the lower-bound atomic guards $d \lleq x$ in $g'$
  let $g'_u$ be the upper-bound atomic guards $x \lleq d$ in $g'$
  return $(g'_l,g'_u,R)$

function backwardLU($Z$, $g_l$, $g_u$, $R$, $LU$, $X_L$, $X_U$):
  let $\sigma$ := $Z \xra{g_l} Z' \xra{g_u;R} Z''$
  update $LU$, $X_L$ and $X_U$ applying £\texttt{newbounds}£ on $Z' \xra{g_u;R} Z''$
  update $LU$, $X_L$ and $X_U$ applying £\texttt{newbounds}£ on $Z \xra{g_l} Z'$
  return $(LU, X_L, X_u)$
\end{lstlisting}



\end{document}